\documentclass[aps,pra,twocolumn,superscriptaddress,amsthm,nopreprintline]{revtex4-2}

\usepackage{amsmath, amssymb, amsthm}
\usepackage{braket}
\usepackage{graphicx}
\usepackage{xcolor}
\usepackage{enumitem}
\usepackage[hidelinks]{hyperref}

\newcommand{\CZ}{\mathrm{CZ}}
\newcommand{\CNOT}{\mathrm{CNOT}}
\newcommand{\Phiplus}{\ket{\Phi^+}}
\newcommand{\ketbra}[2]{\ket{#1}\!\bra{#2}}
\newcommand{\GKtwo}{\ket{G_{K_2}}}
\newcommand{\LFour}{\ket{L_4}}
\newcommand{\GHZn}[1]{\ket{\mathrm{GHZ}_{#1}}}
\newcommand{\Fstar}{F^*}
\newcommand{\xor}{\oplus}
\newcommand{\PQW}{\textsc{pqw}}

\newtheorem{theorem}{Theorem}
\newtheorem{lemma}[theorem]{Lemma}
\newtheorem{proposition}[theorem]{Proposition}
\newtheorem{corollary}[theorem]{Corollary}
\newtheorem{definition}[theorem]{Definition}
\newtheorem{remark}[theorem]{Remark}

\allowdisplaybreaks
\begin{document}

\title{Single-Round Graph-State Distribution from Bell Pairs:
       Exact Fidelity under Resource-Qubit Noise}

\author{Soumyojyoti Dutta}
\affiliation{Indian Institute of Technology Jodhpur,
             Jodhpur 342030, India}
\email{m24iqt014@alumni.iitj.ac.in}

\date{24 April, 2026}

\begin{abstract}
Distributing arbitrary graph states across quantum networks is a central
challenge for modular quantum computing and measurement-based quantum
communication. I present a protocol that converts $|E|$ elementary
two-qubit resource states --- one per edge of the target graph
$G=(V,E)$ --- into the distributed graph state $\ket{G}$ using only
operations local to each party --- single-qubit gates, CZ gates between a
party's own data and resource qubits, and single-qubit measurements ---
together with classical communication; no joint operation between spatially
separated parties is required at any point, and the protocol reaches
arbitrary graph topologies rather than the GHZ class to which prior
walk-based schemes are confined. The elementary step is obtained from the
coined discrete-time quantum walk by replacing the position-permuting shift
operator with a diagonal conditional phase (CZ) gate, and I refer to the
resulting construction as the \emph{phase quantum walk} (PQW); on each edge
it realises the entanglement-assisted non-local CZ gate. Because every gate
and correction is diagonal, all edges run in a single round with the local
correction $C_v=Z_v^{g_v}$, $g_v=\bigoplus_{e\ni v}s_{e,\bar{v}}$, verified
exhaustively for every connected graph on at most four vertices and every
connected graph on five vertices with at most six edges --- 22 graphs,
32{,}212 measurement outcomes --- at $F=1.0$. The central result is an exact
fidelity theory for this protocol under independent Pauli noise on the
resource qubits, with a separate channel on each. When the noise acts after
the local gate, the fidelity factorises over vertices; for depolarising noise
$\Fstar=\prod_v[1+(1-4p/3)^{\deg v}]/2$, with the degree sequence entering
only at second order through $\sum_v\binom{\deg v}{2}$. When it acts on the
stored pair, as link-generation, transmission and storage noise does, it
becomes an edge-cover partition function,
$\Fstar=2^{-|V|}\sum_{S\subseteq V}(1-4p/3)^{2|E(S)|}$, bounded above by the
vertex product and with a leading-order loss coefficient $3/2$ times larger.
Dephasing gives the same fidelity at both placements.
On \texttt{ibm\_kingston} (IBM Heron~r2, CZ-native) I certify the
distributed states by direct fidelity estimation, obtaining
$F=0.815(3)$ for $\LFour$, $0.776(3)$ for $\GHZn{4}$, and $0.733(3)$
for $\ket{C_4}$ in a single native-layout batch. I further establish a
\emph{necessary} condition for native (SWAP-free) execution on a device
of coupling girth $\gamma$: the target graph must have girth at least
$\gamma/3$ and maximum degree at most~$3$. On heavy-hex hardware
($\gamma=12$) this forbids triangles and places the canonical interaction
graphs $I(P_4)$, $I(C_4)$ and $I(K_4)$ in a strict embeddability hierarchy
--- flexible, rigid, and non-embeddable respectively --- with the
measured fidelities descending from the flexible path to the rigid cycle.
The obstruction constrains the circuit, not the state: since
$K_4\sim_{\rm LC}K_{1,3}$, the state $\ket{K_4}$ itself remains reachable
from three resource pairs on the natively embeddable star layout.
\end{abstract}


\maketitle

\section{Introduction}
\label{sec:intro}

Quantum networks require shared entanglement between spatially separated
parties as their fundamental resource. The topology of this entanglement
determines which distributed quantum tasks can be performed: Bell pairs enable
quantum teleportation~\cite{bennett1993teleporting} and quantum key
distribution~\cite{ekert1991quantum}; GHZ states enable multi-party secret
sharing~\cite{briegel1998quantum}; and general graph states enable
measurement-based quantum computation (MBQC)~\cite{raussendorf2001one,browne2005resource},
quantum error-correcting codes distributed across network
nodes~\cite{gottesman1997stabilizer}, and multi-party communication tasks
that neither Bell pairs nor GHZ states can support~\cite{hein2004multiparty}.
As a concrete illustration: in Section~\ref{sec:usecase}, I show
that a four-qubit linear cluster state $\LFour$ distributed by the PQW
enables a remote party to receive an arbitrary single-qubit rotation
$R_x(-\theta)\ket{0}$ from a spatially separated sender using only
single-qubit measurements and 9~classical bits --- no direct quantum
channel between sender and receiver is required after the initial
resource distribution. This end-to-end example grounds the protocol
in an operational task: the PQW is not merely a method for generating
graph states, but a resource-generation layer on which distributed
measurement-based quantum computation can be built.

The PQW operates at the \emph{conversion layer} of a quantum network
stack, above the link layer.
Link-layer protocols --- such as DLCZ~\cite{duan2001long} or
Barrett--Kok~\cite{barrett2005efficient} --- generate Bell pairs
$\Phiplus$ between adjacent nodes as the network primitive; since
$\GKtwo = (H\otimes I)\Phiplus$
(Proposition~\ref{prop:shift_topology}(i)),
a single local Hadamard on one qubit converts each pair to the
PQW resource state at negligible cost.
The PQW takes exactly $|E|$ such resource pairs as input, one per edge
of the target graph $G=(V,E)$, and outputs a single copy of $\ket{G}$
distributed across the $|V|$ parties using only operations local to each
party and classical feed-forward.
For tree graphs, this resource count is optimal: removing an edge
$(u,v)\in E$ splits the tree into two components, defining a bipartition
across which $\ket{G}$ carries exactly one ebit of entanglement, and ---
since link-layer pairs connect adjacent nodes --- the pair on $(u,v)$ is
the only resource crossing this cut. Because LOCC cannot increase
entanglement across any cut~\cite{horodecki2009quantum}, at least one
pair per edge is necessary.
For graphs with cycles this per-edge cut argument no longer applies:
the entanglement of $\ket{G}$ across a bipartition equals the
$\mathbb{F}_2$ rank of the corresponding off-diagonal block of the
adjacency matrix, which can be strictly smaller than the number of
crossing edges (for $C_3$, the single-vertex cut is crossed by two
edges but carries only one ebit).
For such graphs $|E|$ pairs is achievable but not in general optimal.
Because local Clifford operations are free, $\ket{G}$ can be distributed
from the pairs of any LC-equivalent graph whose edges are network links:
$\ket{C_3}$ needs only the two pairs of the path $A$---$B$---$C$, and
$\ket{K_4}$ only the three of $K_{1,3}$ (Sec.~\ref{sec:embeddability}),
against $|V|-1$ pairs required by connectivity alone. The relevant count is
that of a minimum edge representative of the LC orbit, which is known for
graph states of up to 12 qubits~\cite{cabello2011optimal} and has been
studied for larger orbits~\cite{sharma2026minimising}.
The PQW thus establishes a clean separation of concerns:
link-layer protocols are responsible for generating one elementary pair
per edge of the chosen graph representative; the PQW is responsible for
converting them into the target graph state with no additional entanglement
overhead.

Graph-state distribution over networks of bipartite links has been studied
outside the walk framework since Cuquet and
Calsamiglia~\cite{cuquet2012growth}, with protocols optimised for Bell-pair
count, rounds and noise~\cite{meignant2019distributing,hahn2019quantum,fischer2021distributing,sen2025multipartite,huang2024spacetime};
these are compared in Sec.~\ref{sec:discussion}.
Quantum walks provide a natural mechanism for entanglement distribution.
Meignant et al.~\cite{meignant2019distributing} showed that discrete-time
quantum walk (DTQW) steps can distribute Bell pairs and GHZ states using
elementary two-qubit resources. Chen et al.~\cite{chen2025entanglement}
extended this to a quantum repeater framework, distributing $d$-dimensional
Bell and GHZ states across arbitrary network topologies --- selected by a
Steiner-tree search --- with experimental demonstrations on superconducting
hardware. In both cases the shift is a (generalised) conditional permutation
of the position register, reducing to CNOT for qubits, and both distribute
GHZ-class entanglement: their fusion primitives merge GHZ-class resources
into larger GHZ states, so neither \emph{walk-based} protocol reaches general
graph topologies. (Arbitrary graph states are reached by a separate,
non-walk protocol of Meignant et al., discussed in
Sec.~\ref{sec:discussion}, at the cost of joint two-qubit measurements at
intermediate nodes.) I emphasise that
``arbitrary networks'' in this line of work refers to the network over which
the state is distributed, not to the class of distributed states.

In this paper I present a graph-state distribution protocol built on two
ingredients. The first is a diagonal conditional-phase (CZ) step in place of
the position-permuting shift of prior walk-based schemes; the second is a
two-register architecture in which every party holds a data qubit alongside
its resource qubits. The two play distinct roles. The CZ form is symmetric
(unlike CNOT) and is the natural gate of graph-state preparation
(Eq.~\eqref{eq:graphstate}), so the target state appears in its own frame
and the required corrections reduce to $Z$ operators; the CZ and CNOT shifts
are themselves locally equivalent, and the choice between them is a change of
frame rather than of capability
(Proposition~\ref{prop:shift_topology}). It is the separate data register
that allows the adjacency of an arbitrary target graph to be imprinted edge
by edge using strictly local operations
(Remark~\ref{rem:reachable_class}). Because the elementary step is obtained
from the coined DTQW by the single substitution of a phase for a permutation,
I refer to the construction as the \emph{phase quantum walk} (PQW), the name
recording that derivation rather than a taxonomic claim
(Sec.~\ref{sec:pqw}). I prove that the two ingredients together provide a
unified framework for distributing \emph{arbitrary} graph states across
quantum networks from elementary two-qubit resources.

The main contributions of this paper are:
\begin{enumerate}[label=(\roman*),leftmargin=*,itemsep=2pt]
  \item A \textbf{graph-state distribution primitive} (the \emph{phase
    quantum walk}, PQW): a two-register construction obtained from the coined
    DTQW by replacing the position-shift with a conditional phase (CZ), whose
    structural properties --- symmetry, $Z$-transparency, and local
    equivalence to the CNOT shift --- are established
    (Sec.~\ref{sec:pqw}).
  \item The \textbf{Byproduct Lemma}: a PQW step teleports edge entanglement
    with a correctable Pauli $X$ byproduct (Lemma~\ref{lem:byproduct}).
  \item A complete \textbf{distribution protocol} for the four-qubit linear
    cluster state $\LFour$ with analytical correction formula and $F=1.0$
    verification for all 64 outcomes (Sec.~\ref{sec:protocol}).
  \item \textbf{Exact closed-form fidelities} under independent resource-qubit
    noise at both circuit placements: a product over vertices weighted by
    degree when noise acts after the local gate, and an edge-cover partition
    function when it acts on the stored pair, with dephasing
    placement-independent (Sec.~\ref{sec:noise}).
  \item \textbf{Generalisation} to arbitrary graphs
    (Sec.~\ref{sec:generalisation}), with the universal formula verified
    exhaustively on 22 connected graphs (32{,}212 outcomes) and the
    topology-specific formulas on eight
    (Tables~\ref{tab:z_only_verify} and~\ref{tab:topologies}).
  \item \textbf{LC-inequivalence} of the protocol outputs, verified via
    Schmidt rank --- a known classification fact included as certification
    that the protocol distributes entanglement beyond the GHZ class
    (Sec.~\ref{sec:lc_ineq}).
  \item The \textbf{Universal Z-only Correction Theorem}: a single correction
    formula $C_v = Z_v^{g_v}$ valid for all graphs and all outcomes, proved
    via stabiliser tracking (Sec.~\ref{sec:universal_correction}).
  \item A \textbf{concrete distributed use case}: PQW-distributed $\LFour$
    enables remote single-qubit state preparation via MBQC, using only
    single-qubit measurements and classical communication, verified at
    $F=1.0$ for all outcomes (Sec.~\ref{sec:usecase}).
  \item \textbf{Hardware execution and an embeddability theorem}: direct
    fidelity estimation on \texttt{ibm\_kingston} certifies three protocols
    at $F=0.73$--$0.82$, and a girth argument gives necessary conditions
    for native SWAP-free execution on a device of coupling girth $\gamma$
    --- girth at least $\gamma/3$ and maximum degree at most $3$, which on
    heavy-hex ($\gamma=12$) forbids triangles --- under which the canonical
    interaction graph $I(K_4)$ is provably non-embeddable, although the state
    $\ket{K_4}$ itself remains reachable through its LC-equivalent star
    representative
    (Sec.~\ref{sec:hardware},~\ref{sec:embeddability}).
\end{enumerate}

\section{Background}
\label{sec:background}

\subsection{Graph States and Stabiliser Formalism}

\begin{definition}[Graph state~\cite{hein2004multiparty}]
\label{def:graphstate}
For a graph $G=(V,E)$, the associated graph state is
\begin{equation}
  \ket{G} = \prod_{(u,v)\in E} \CZ_{uv}\,\ket{+}^{\otimes|V|}.
  \label{eq:graphstate}
\end{equation}
It is the unique $+1$ eigenstate of the stabiliser generators
$K_v = X_v\!\prod_{u\sim v} Z_u$ for each $v\in V$.
\end{definition}

The simplest graph state is the two-qubit graph state
$\GKtwo = \CZ\ket{++} = \frac{1}{2}(\ket{00}+\ket{01}+\ket{10}-\ket{11})$,
which serves as the elementary resource throughout this paper.

The four-qubit \emph{linear cluster state} $\LFour$ is the graph state on the
path $P_4 = A$---$B$---$C$---$D$, with stabilisers
$K_A = X_AZ_B$, $K_B = Z_AX_BZ_C$, $K_C = Z_BX_CZ_D$, $K_D = Z_CX_D$.

\subsection{Walk-Based Distribution and the Origin of the GHZ Restriction}

A coined DTQW~\cite{aharonov1993quantum} (see
Ref.~\cite{kadian2021review} for a survey of walk formulations and their
application domains) on $\mathcal{H}_C\otimes\mathcal{H}_P$ applies
$U = S\cdot(C\otimes I)$ per step, where $C$ is a coin and $S$ a shift.
The CNOT shift satisfies $S_\mathrm{CNOT}\ket{c}\ket{x} = \ket{c}\ket{x\oplus c}$
and prepares Bell pairs via $\CNOT\ket{+0} = \Phiplus$. In the walk-based
distribution protocols built on this step, the entangled resource particles
are themselves the walk registers and the output: a step fuses two GHZ-class
resources at a node into a single larger GHZ-class state, so iterating over a
network yields GHZ states among the selected parties. The restriction to
GHZ-class outputs is therefore a property of this architecture and of the
resources it consumes; it is not a consequence of the shift operator's
identity, since the CNOT and CZ shifts are locally equivalent
(Proposition~\ref{prop:shift_topology}).

\paragraph{Architectural distinction from prior DTQW protocols.}
In the protocol of Chen et al.~\cite{chen2025entanglement}, the resource particles
serve \emph{dual roles}: they are simultaneously the walk registers (coin
and position) and the output qubits --- the output entangled state is created
\emph{within} the walk register via entanglement swapping, with no separate
system required.

The PQW takes a fundamentally different approach. It introduces explicit
\emph{data qubits} $d_v$ at each party, separate from the resource qubits
$r_{v,e}$. The CZ walk step acts \emph{between} data qubit and resource
qubit; the resource qubits are subsequently measured and discarded; the
data qubits hold the output graph state. The step therefore does not act
on the data qubits directly; it acts on the resource register and
teleports graph-state entanglement into the data register via
measurement-induced byproducts.

The role of this separation is to decouple two functions that the
single-register protocols must combine. When the resource particles are also
the output, the systems that a node measures are drawn from the same pool
that must survive to carry the distributed state, and fusing GHZ-class
resources at a node yields a GHZ-class output. With a data register the two
functions come apart: the resource qubits are measured and discarded, while
the data qubits --- never measured during distribution --- accumulate the
adjacency of the target graph one edge at a time. This is not the only route
to arbitrary graph states; the edge-decorated-complete-graph protocol of
Meignant \emph{et al.}~\cite{meignant2019distributing} reaches them with no
data register, at the cost of Bell-state measurements at intermediate nodes
(Remark~\ref{rem:reachable_class}). What the data register buys is arbitrary
target graphs using operations that are \emph{strictly local to each party}:
single-qubit gates, CZ between a party's own data and resource qubits, and
single-qubit measurements.

\textit{Relationship to graph state preparation and MBQC.}
The phase quantum walk should not be confused with two related but distinct
frameworks.

\emph{(a) Local graph state preparation.}
Standard graph state preparation~\cite{hein2004multiparty} applies CZ gates
directly between co-located data qubits. The PQW operates in the quantum
\emph{network} setting: data qubits are held by spatially separated parties
who share only pre-distributed two-qubit resource states $\GKtwo$ and
classical communication. No direct interaction between data qubits is
possible. The CZ gate appears between a data qubit and a \emph{local}
resource qubit at the same party --- never between two data qubits.

\emph{(b) Measurement-based quantum computation (MBQC).}
In MBQC~\cite{raussendorf2001one}, a graph state is the \emph{input}:
it must be fully prepared before computation begins, and is then consumed
by adaptive single-qubit measurements. The PQW produces a graph state
as its \emph{output}: the protocol creates it distributedly across the
network from elementary two-qubit resources. The PQW is therefore
the \emph{resource generation layer} of the kind that must precede MBQC in a
distributed quantum computing architecture.

\section{The Phase Quantum Walk}
\label{sec:pqw}

\begin{definition}[Phase Quantum Walk]
\label{def:pqw}
A \emph{phase quantum walk} on $\mathcal{H}_P\otimes\mathcal{H}_C$
(position $\otimes$ coin) has single-step evolution
\begin{equation}
  U_\varphi = (I_P\otimes H_C)\cdot\CZ_{PC},
  \label{eq:pqw_step}
\end{equation}
where the \emph{conditional phase operator} is
\begin{equation}
  \CZ_{PC} = \mathrm{diag}(1,\,1,\,1,\,-1).
  \label{eq:phase_op}
\end{equation}
The operator $\CZ_{PC}$ is diagonal and applies phase $(-1)^{P\cdot C}$.
The position register is \emph{never permuted}: the step acts by conditional
phase followed by a coin Hadamard, and its effect is entirely a
redistribution of amplitude within the fixed position basis. In the
distribution protocol of Sec.~\ref{sec:protocol} the position register is
realised by a party's data qubit and the coin register by the resource qubit
it holds, so that the step of Eq.~\eqref{eq:pqw_step} is the pair
$\CZ(d_v,r_{v,e})$ followed by $H(r_{v,e})$.
\end{definition}

\begin{remark}[Scope of the novelty claim]
\label{rem:novelty_scope}
What is claimed as new is the construction as a whole: the two-register
architecture, the measurement pattern it induces, and the universal $Z$-only
correction that follows from them for arbitrary target graphs
(Theorem~\ref{thm:universal_z}). The protocol is executable in a single
distribution round under the parallelisation principle established by Fischer
and Towsley~\cite{fischer2021distributing}; that principle is not claimed
here. Two further ingredients are
explicitly \emph{not} claimed. The shift substitution taken in isolation is a
change of local frame (Proposition~\ref{prop:shift_topology}). And the
stabiliser-absorption mechanism that explains why the residual $X$ corrections
of Theorems~\ref{thm:correction} and~\ref{thm:c4_correction} may be dropped
(Remark~\ref{rem:relation_topology}) was identified independently, in a
different protocol context, by Zheng \emph{et
al.}~\cite{zheng2026scalable} (Sec.~\ref{sec:discussion}); the proof of
Theorem~\ref{thm:universal_z} does not use it. Phases appear in two other families of walk
model, which should be distinguished from this one. In phase-disordered coined
walks~\cite{schreiber2011decoherence} the phases are a randomised realisation
applied to a transport walk, and the object of study is localisation rather
than a prepared output state. In the graph-phased Szeg\H{e}dy walk of Ortega
and Mart\'{\i}n-Delgado~\cite{ortega2025complex}, which equips Szeg\H{e}dy's
model with link phases and local arbitrary phase rotations, the phases
decorate a walk on the edge space of a graph, the two node registers being the
walk itself, and the target application is search. Neither pairs a diagonal
conditional-phase step with a separate data register, and neither is used for
graph state distribution. The link phases of
Ref.~\cite{ortega2025complex} are nevertheless the closest existing analogue
of the variational $\mathrm{CP}(\theta)$ extension proposed as Open
Problem~(ii) in Sec.~\ref{sec:discussion}.
\end{remark}

\paragraph{Origin and status of the ``walk'' terminology.}
The PQW is derived from the coined DTQW by a single structural
substitution, and the name records that lineage rather than a taxonomic
claim. The starting observation is the one recorded in
Sec.~\ref{sec:background}: the walk-based repeater protocols of Chen
\emph{et al.}~\cite{chen2025entanglement} distribute only GHZ-class states.
Seeking a construction whose elementary step produces the
$\ket{G_{K_2}}$ correlations of a graph state directly, rather than the
$ZZ$ correlations of a Bell pair, I replace the position-permuting CNOT
shift with the diagonal CZ: where CNOT permutes the position register
conditioned on the coin, CZ instead applies a conditional phase,
$\ket{c}\ket{x}\mapsto(-1)^{cx}\ket{c}\ket{x}$, leaving positions fixed. This
substitution --- permutation replaced by phase --- is the definition of the
PQW step~\eqref{eq:pqw_step}, and it is the sole change from the standard
coined-walk template. I stress that the substitution is a change of
\emph{frame}, not of capability: the two shifts are locally equivalent
(Proposition~\ref{prop:shift_topology}), and it is the two-register
architecture introduced below, not the gate identity, that makes arbitrary
target graphs reachable without any joint operation between separated parties
(Remark~\ref{rem:reachable_class}). What the CZ form buys is that the target
graph state appears in its natural frame, so the required corrections are
$Z$-only.

I do not claim that the resulting object is a quantum walk in the strict
dynamical sense. The position register is never permuted, and the
distribution protocol applies the step exactly once to each of the $2|E|$
resource qubits (Definition~\ref{def:general_protocol}), never iterating it
on a fixed position--coin pair. The data qubit $d_v$ is shared by the
$\deg(v)$ steps at party $v$, but those steps use distinct coin qubits
$r_{v,e}$, each used once and then measured; the $\deg(v)$ gates
$\CZ(d_v,r_{v,e})$ are diagonal and mutually commuting, and no
non-diagonal operation acts on $d_v$ during distribution, so their
composition is a single diagonal unitary rather than an iterated walk.
There is consequently no transport, spreading, or multi-step spatial
interference. What the substitution provides is the correct \emph{gate-level}
ingredient --- a symmetric, $Z$-transparent, $X$-basis-correlating conditional
operator (Remark~\ref{rem:cz_props}, Propositions~\ref{prop:xbasis} and
\ref{prop:output}) --- from which the graph-state distribution protocol and its
universal correction (Theorem~\ref{thm:universal_z}) follow. The remainder of
the paper concerns those consequences; the walk framing is the route by which
they were found, not a load-bearing claim.

I record two elementary properties of $\CZ$ that are used repeatedly,
then establish the structural propositions underpinning the protocol.

\begin{remark}[Elementary properties of $\CZ$]
\label{rem:cz_props}
Since $\CZ = \mathrm{diag}(1,1,1,-1)$ is diagonal in the computational basis:
\begin{enumerate}[label=(\roman*)]
  \item \textbf{Z-transparency.}
    $[\CZ,\,Z\otimes I] = [\CZ,\,I\otimes Z] = [\CZ,\,Z\otimes Z] = 0$,
    since all operators are simultaneously diagonal.
  \item \textbf{Symmetry.}
    $\CZ_{PC} = \CZ_{CP}$, since the phase $(-1)$ is applied if and only
    if both qubits are $\ket{1}$ --- a symmetric condition.
\end{enumerate}
These follow directly from the definition and require no further proof.
\end{remark}

\begin{proposition}[$X$-basis equivalence]
\label{prop:xbasis}
$(I_P\otimes H_C)\cdot\CZ\cdot(I_P\otimes H_C) = \CNOT_{P\to C}$ and
$(H_P\otimes I_C)\cdot\CZ\cdot(H_P\otimes I_C) = \CNOT_{C\to P}$.
\end{proposition}
\begin{proof}
Since $\CZ = \ketbra{0}{0}\otimes I + \ketbra{1}{1}\otimes Z$ and $HZH=X$,
\begin{align*}
  &(I\otimes H)\,\CZ\,(I\otimes H)\\
  &\quad= \ketbra{0}{0}\otimes I + \ketbra{1}{1}\otimes HZH\\
  &\quad= \ketbra{0}{0}\otimes I + \ketbra{1}{1}\otimes X = \CNOT_{P\to C}.
\end{align*}
The second identity follows by symmetry of $\CZ$ (Remark~\ref{rem:cz_props}(ii)).
\end{proof}

\begin{proposition}[Graph state output]
\label{prop:output}
$\CZ\ket{++} = \GKtwo$ with stabilisers $K_P = X_PZ_C$ and $K_C = Z_PX_C$.
\end{proposition}
\begin{proof}
Direct calculation: $\CZ\ket{++} = \tfrac{1}{2}(\ket{00}+\ket{01}+\ket{10}-\ket{11})$.
Stabiliser verification is immediate.
\end{proof}

\begin{proposition}[Local equivalence of the shift]
\label{prop:shift_topology}
\begin{enumerate}[label=(\roman*),leftmargin=*]
  \item The two shift operators and their elementary resources are related by
    local Hadamards:
    \[
      \CZ = (I\otimes H)\,\CNOT\,(I\otimes H),\qquad
      \GKtwo = (I\otimes H)\Phiplus .
    \]
  \item Consequently the PQW step --- $\CZ(d,r)$, then $H(r)$, then a
    $Z$-basis measurement of $r$ --- is \emph{identical} to the CNOT form in
    which $H(r)$ is applied before a $\CNOT(d\!\to\!r)$ and $r$ is measured
    directly. The two are the same protocol written in different local frames,
    so the choice of shift operator alone cannot determine which states a
    protocol is able to distribute.
  \item What the CZ choice does fix is the \emph{frame}. With $\GKtwo$
    resources and the CZ step, the post-measurement data stabilisers are
    $(-1)^{g_v}K_v$ (Lemma~\ref{lem:phase}) with $\ket{G}$ itself --- not a
    local-Clifford image of it --- as the target, so the corrections are
    $Z$-only (Theorem~\ref{thm:universal_z}).
\end{enumerate}
\end{proposition}
\begin{proof}
(i) Direct computation: $\CZ = \ketbra{0}{0}\otimes I + \ketbra{1}{1}\otimes Z$
and $HZH = X$ give the first identity; applying $I\otimes H$ to
$\GKtwo=\CZ\ket{++}$ gives $\CNOT\ket{+0}=\Phiplus$.
(ii) Substituting $\CZ(d,r) = H_r\,\CNOT(d\!\to\!r)\,H_r$ into the step and
using $H_r^2 = I$ moves the Hadamard from after the gate to before it; the
measurement statistics and post-measurement data states are unchanged
(verified numerically for all outcomes).
(iii) Immediate from Lemma~\ref{lem:phase} and
Theorem~\ref{thm:universal_z}.
\end{proof}

\begin{remark}[What determines the reachable class]
\label{rem:reachable_class}
Proposition~\ref{prop:shift_topology}(ii) shows that the distributable state
class is a property of the protocol architecture and of the operations
permitted at a node, not of the shift operator in isolation. The three
families in the literature separate along exactly these lines. In the
walk-based repeater protocols of Chen \emph{et al.}~\cite{chen2025entanglement},
the resource particles serve simultaneously as walk registers and as output
qubits; the fusion primitive merges GHZ-class resources (qudit Bell and GHZ
states) into larger GHZ states, and the reachable class is correspondingly
GHZ/star. The edge-decorated-complete-graph protocol of Meignant \emph{et
al.}~\cite{meignant2019distributing} reaches arbitrary graph states with no
data register, but requires Bell-state measurements --- joint two-qubit
operations spanning two parties --- at intermediate nodes. The PQW introduces
a separate data register per party, which allows the target graph's adjacency
to be imprinted edge by edge, and thereby reaches arbitrary graph states with
\emph{every two-qubit gate confined within a single party}: each CZ acts on a
data qubit and a co-located resource qubit held by the same node, and the
measurements are single-qubit. That combination ---
arbitrary target graphs with no joint operation between separated parties ---
is what distinguishes the present protocol, and it follows from the
architecture rather than from the gate identity.
\end{remark}

\begin{remark}[Properties of the step unitary]
\label{rem:step_unitary}
The step unitary $U_\varphi=(I_P\otimes H_C)\CZ_{PC}$ is block-diagonal in
the computational basis, with eigenvalues $\{1,-1,e^{i\pi/4},e^{-i\pi/4}\}$
and minimal period $U_\varphi^8=I$ (verified numerically,
$\max|U_\varphi^8-I|<10^{-15}$). I record these properties of the gate for
completeness; the distribution protocol applies the step exactly once per
resource qubit, $2|E|$ times in total (Definition~\ref{def:general_protocol}),
and no iterated evolution $U_\varphi^k$ with $k>1$ occurs anywhere in the
protocol or in the results that follow.
\end{remark}

\begin{lemma}[Phase Walk Entanglement Transfer]
\label{lem:byproduct}
Let data qubit $d$ be in $\ket{+}$, and $(r,r')$ be in $\GKtwo$.
After $\CZ(d,r)$, $H(r)$, and measuring $r$ with outcome $s\in\{0,1\}$,
the pair $(d,r')$ is in $(X^s\otimes I)\,\Phiplus$,
with each outcome occurring with probability $\tfrac{1}{2}$.
\end{lemma}
\begin{proof}
Expand $\ket{+}_d\otimes\GKtwo_{r,r'}$:
\begin{align*}
  &\ket{+}\otimes\tfrac{1}{2}(\ket{00}+\ket{01}+\ket{10}-\ket{11})\\
  &\xrightarrow{\CZ}
  \tfrac{1}{2\sqrt{2}}\bigl[\ket{0}(\ket{00}+\ket{01}+\ket{10}-\ket{11})\\
  &\hphantom{\xrightarrow{\CZ}\tfrac{1}{2\sqrt{2}}\bigl[}
   +\ket{1}(\ket{00}+\ket{01}-\ket{10}+\ket{11})\bigr]\\
  &\;=\tfrac{1}{\sqrt{2}}\bigl[\ket{+}_r\,(X^0\otimes I)\Phiplus_{d,r'}
   +\ket{-}_r\,(X^1\otimes I)\Phiplus_{d,r'}\bigr]\\
  &\xrightarrow{H_r}
  \tfrac{1}{\sqrt{2}}\bigl[\ket{0}_r\,(X^0\otimes I)\Phiplus_{d,r'}
   +\ket{1}_r\,(X^1\otimes I)\Phiplus_{d,r'}\bigr].
\end{align*}
Measuring $r$ in the $Z$ basis (i.e.\ in the $X$ basis before $H$) with
outcome $s\in\{0,1\}$ leaves $(X^s\otimes I)\Phiplus_{d,r'}$, with probability
$1/2$.
\end{proof}

\paragraph{Relation to non-local gates.}
Applied at both endpoints of an edge $e=(u,v)$, the two steps followed by
the corrections $Z_u^{s_{e,v}}Z_v^{s_{e,u}}$ implement $\CZ(d_u,d_v)$ exactly,
on \emph{arbitrary} input states of the data qubits (verified on all four
outcomes), at a cost of one ebit and one classical bit in each direction.
This is the entanglement-assisted non-local controlled gate of Eisert
\emph{et al.}~\cite{eisert2000optimal}, who prove that cost necessary and
sufficient; related constructions are the cat-entangler and
cat-disentangler of Yimsiriwattana and
Lomonaco~\cite{yimsiriwattana2004generalized} and gate
teleportation~\cite{gottesman1999demonstrating}. The edge-level operation is
therefore not new; what the remainder of the paper studies is its parallel
composition over all edges of a target graph, the correction structure that
results, its behaviour under resource-qubit noise, and its native embedding on
hardware.

\begin{remark}[Single step output is a Bell pair, not a graph state]
\label{rem:bell_vs_graph}
The output of a single PQW step is a \emph{Bell pair} $\Phiplus$ (up to
a correctable Pauli $X^s$), not the graph state $\GKtwo = \CZ\ket{++}$.
Although $\Phiplus$ and $\GKtwo$ are LC-equivalent (related by local
Hadamards), they are distinct states. This distinction is important:
the target graph state $\ket{G}$ does \emph{not} emerge from a single
walk step but from the \emph{combination} of $2|E|$ walk steps, two per
edge (one at each endpoint), followed by local corrections. On an edge
$e=(u,v)$ the step at $u$ leaves $(d_u,r_{v,e})$ in a correctable Bell pair
(Lemma~\ref{lem:byproduct}); the step at $v$ then transfers that
correlation onto $(d_u,d_v)$, so each \emph{pair} of steps contributes one
edge of the target graph, and the graph state structure accumulates across
all edges via the stabiliser propagation tracked in Lemma~\ref{lem:phase}.
\end{remark}

\section{The $\boldsymbol{\LFour}$ Distribution Protocol}
\label{sec:protocol}

\subsection{Setup and Circuit}

Four parties $A$, $B$, $C$, $D$ are connected along the path
$P_4 = A$---$B$---$C$---$D$. For each edge $e=(u,v)$, a resource state
$\GKtwo$ is prepared and distributed: $r_{u,e}$ goes to $u$, $r_{v,e}$
to $v$. Total qubit count: $4 + 2\times3 = 10$ qubits.

\begin{figure}[t]
  \centering
  \includegraphics[width=\columnwidth]{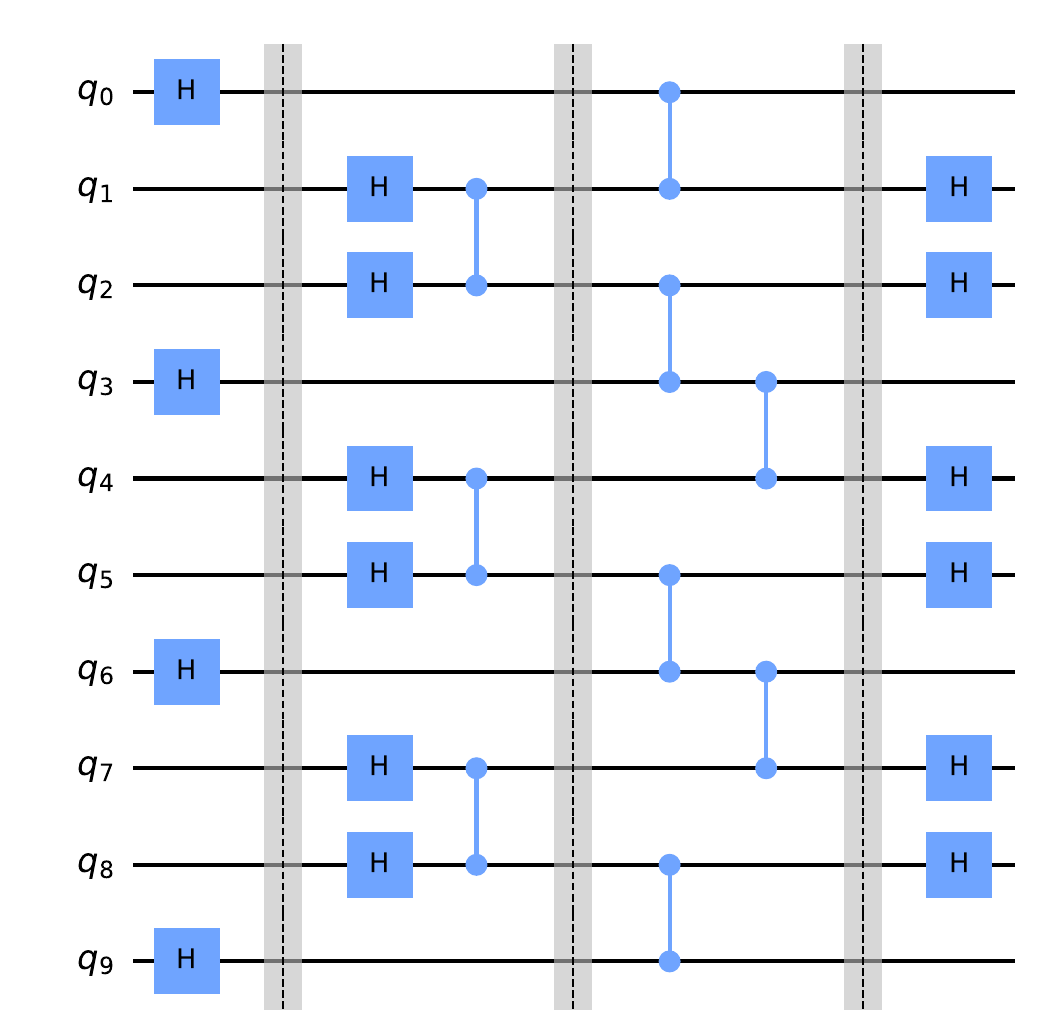}
  \caption{Qiskit circuit for the $\LFour$ distribution protocol (10 qubits).
    Barriers separate the four stages: (S1) Hadamard on data qubits,
    (S2) resource state preparation, (S3) phase walk CZ gates,
    (S4) Hadamard on resource qubits before measurement.}
  \label{fig:circuit}
\end{figure}

\subsection{Protocol}

\begin{enumerate}[label=\textbf{S\arabic*.},leftmargin=*,itemsep=1pt]
  \item \textbf{Initialise.} $d_v \leftarrow \ket{+}$ for $v\in\{A,B,C,D\}$.
  \item \textbf{Resources.} Prepare $\GKtwo_{r_{u,e},r_{v,e}}$ per edge $e$.
  \item \textbf{Phase walk.} Each party $v$ applies $\CZ(d_v, r_{v,e})$,
    $H(r_{v,e})$, measures $r_{v,e}\to s_{v,e}$ for all adjacent $e$.
  \item \textbf{Correction.} Broadcast outcomes; apply Thm.~\ref{thm:correction}.
\end{enumerate}

Label the six resource outcomes: $s_1$ (A-side of $AB$), $s_2$ (B-side),
$s_3$ (B-side of $BC$), $s_4$ (C-side), $s_5$ (C-side of $CD$), $s_6$ (D-side).

\subsection{Correction Formula}

\begin{theorem}[Correction formula for $\LFour$ distribution]
\label{thm:correction}
After the following local Pauli corrections, data qubits are in $\LFour$ for
all $64$ measurement outcomes:
\begin{align}
  A &: I, \label{eq:corr_A}\\
  B &: X^{s_2}, \label{eq:corr_B}\\
  C &: X^{s_1\xor s_4}, \label{eq:corr_C}\\
  D &: X^{s_2\xor s_3\xor s_6}\cdot Z^{s_1\xor s_4\xor s_5}.
  \label{eq:corr_D}
\end{align}
\end{theorem}
\begin{proof}
The proof rests on the Phase Lemma (Lemma~\ref{lem:phase},
Sec.~\ref{sec:universal_correction}), which is established by stabiliser
tracking and does not use this theorem. Lemma~\ref{lem:byproduct} cannot be
applied step by step here: its hypotheses --- data qubit in $\ket{+}$,
resource pair in $\GKtwo$ --- hold for the first step on each edge but not
thereafter (the second $\CZ$ at $B$ acts on a data qubit already entangled
by the first, and the second step on an edge acts on a resource qubit
already paired with the opposite data qubit).

By Lemma~\ref{lem:phase}, the post-measurement data state
$\ket{\psi_\mathbf{s}}$ is stabilised by $(-1)^{g_v}K_v$, with far-side
parities $g_A=s_2$, $g_B=s_1\xor s_4$, $g_C=s_3\xor s_6$, $g_D=s_5$
(Sec.~\ref{sec:notation}). For a Pauli operator $C$, $CK_vC^\dagger
=(-1)^{c_v}K_v$ with $c_v=1$ exactly when $C$ anticommutes with $K_v$, so
$C\ket{\psi_\mathbf{s}}$ is stabilised by $(-1)^{g_v\xor c_v}K_v$ and equals
$\LFour$ up to a global phase if and only if $c_v=g_v$ for every $v$. For the
corrections~\eqref{eq:corr_A}--\eqref{eq:corr_D},
$C=X_B^{s_2}\,X_C^{s_1\xor s_4}\,X_D^{s_2\xor s_3\xor s_6}Z_D^{s_1\xor s_4\xor s_5}$,
and:
\begin{itemize}[leftmargin=*,itemsep=1pt]
  \item $K_A=X_AZ_B$ anticommutes with $X_B^{s_2}$ only, so $c_A=s_2=g_A$;
  \item $K_B=Z_AX_BZ_C$ anticommutes with $X_C^{s_1\xor s_4}$ only, so
    $c_B=s_1\xor s_4=g_B$;
  \item $K_C=Z_BX_CZ_D$ anticommutes with $X_B^{s_2}$ and
    $X_D^{s_2\xor s_3\xor s_6}$, so $c_C=s_3\xor s_6=g_C$;
  \item $K_D=Z_CX_D$ anticommutes with $X_C^{s_1\xor s_4}$ and
    $Z_D^{s_1\xor s_4\xor s_5}$, so $c_D=s_5=g_D$.
\end{itemize}
Hence $C\ket{\psi_\mathbf{s}}$ is stabilised by every $K_v$ and equals
$\LFour$. These corrections are the member
$x=(0,\,g_A,\,g_B,\,g_A\xor g_C)$ of the family of
Corollary~\ref{cor:gauge_family}, obtained by fixing $C_A=I$ and solving
$z_A=z_B=z_C=0$; they differ from the $Z$-only form of
Theorem~\ref{thm:universal_z} by
$\bigotimes_v C_v = \bigl(\bigotimes_v Z_v^{g_v}\bigr)\,
K_B^{\,g_A}K_C^{\,g_B}K_D^{\,g_A\xor g_C}$ up to a global phase, an element
of $\mathrm{Stab}(\LFour)$ (Remark~\ref{rem:relation_topology}). The $Z$-only
form should be preferred in practice (Remark~\ref{rem:practical}).
All 64 outcomes verified at $F=1.0$ ($1-F<10^{-12}$).
\end{proof}

\subsection{Verification}

All 64 outcomes verified at $F=1.0$ by exact statevector simulation
($1-F < 10^{-12}$, Qiskit 2.2.3 \texttt{Statevector}).

\begin{figure}[t]
  \centering
  \includegraphics[width=\columnwidth]{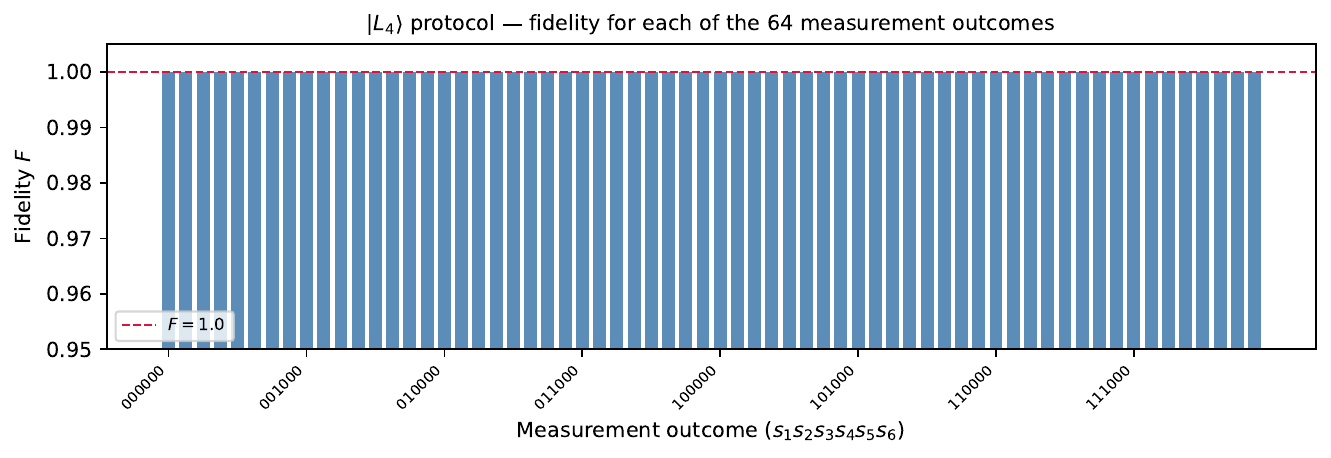}
  \caption{Fidelity with $\LFour$ for each of the 64 measurement outcomes
    after applying corrections of Thm.~\ref{thm:correction}. All 64 bars
    coincide with $F=1.0$ (red dashed). The outcome probabilities are
    separately verified to equal $1/64$ to within $10^{-12}$; they are not
    plotted here.}
  \label{fig:l4_fidelity}
\end{figure}

\section{Noisy Performance}
\label{sec:noise}

\begin{remark}[Noise model]
\label{rem:noise_model}
Theorems~\ref{prop:dep} and~\ref{prop:pd} assume independent
single-qubit noise acting on the $k=2|E|$ resource qubits \emph{after the
$\CZ$ of the walk step and before the pre-measurement Hadamard}; data
qubits, $\CZ$ gates, and readout are taken as noiseless. The placement
matters and is not conventional: a channel acting on the resource qubits
\emph{before} the walk $\CZ$ gives a different result, because
$\CZ\,X_r\,\CZ^\dagger = Z_d X_r$ carries an uncorrectable $Z$ onto the
data qubit, whereas after the $\CZ$ an $X_r$ error only flips a
measurement outcome and is absorbed by the correction (proof of
Theorem~\ref{prop:dep}). The other placement is treated exactly in
Theorem~\ref{thm:pre_cz}. The hardware validation
(Sec.~\ref{sec:hardware}) operates under full device noise and reports
direct fidelity estimation of the corrected data-qubit state rather than
an effective single-qubit rate.
\end{remark}

\begin{lemma}[Error parity]
\label{lem:error_parity}
Let a Pauli error $E=\bigotimes_r E_r$ act on the resource qubits between the
walk $\CZ$ and the pre-measurement Hadamard. For a resource qubit
$r_{e,u}$ held by $u$ on edge $e=(u,v)$, write $\mathrm{fd}(r_{e,u})=v$ for
the endpoint whose correction bit it feeds, so that $s_{e,u}$ enters $g_v$
(Sec.~\ref{sec:notation}). Then for \emph{every} measurement outcome
$\mathbf{s}$ the corrected data state equals $\ket{G}$ if
\begin{equation}
  \bigoplus_{r\,:\,\mathrm{fd}(r)=v}\;
  \bigl[\,E_r\in\{Y,Z\}\,\bigr] \;=\; 0
  \qquad\text{for every } v\in V,
  \label{eq:parity_condition}
\end{equation}
and is orthogonal to $\ket{G}$ otherwise. In particular the fidelity is
independent of $\mathbf{s}$ and takes only the values $0$ and $1$.
\end{lemma}
\begin{proof}
Every operation in the protocol is Clifford and every error is Pauli, so by
Lemma~\ref{lem:phase} the post-measurement data state is stabilised by
$(-1)^{g_v(\mathbf{s})\oplus\delta_v(E)}K_v$, where $\delta_v(E)$ is the sign
flip induced on $K_v$ by $E$. Applying the correction $Z_v^{g_v}$ leaves the
state stabilised by $(-1)^{\delta_v(E)}K_v$. If $\delta(E)=0$ every generator
is restored and the state is $\ket{G}$; if $\delta_v(E)=1$ for some $v$ the
state lies in the $-1$ eigenspace of $K_v$ and is therefore orthogonal to
$\ket{G}$. Hence the fidelity is $0$ or $1$.

Conjugation by a Pauli reverses the sign of exactly those stabiliser
generators it anticommutes with, so signs compose additively and $\delta$ is
$\mathbb{F}_2$-linear: $\delta(E)=\bigoplus_r \delta(E_r)$. It therefore
suffices to evaluate $\delta$ on single-qubit errors at $r=r_{e,u}$.
A $Z_r$ error commutes through the walk $\CZ$ (Remark~\ref{rem:cz_props}(i))
and $HZH^\dagger = X$ makes it an $X$ on the measured qubit, flipping the
recorded outcome $s_{e,u}$; that outcome enters $g_v$ with $v=\mathrm{fd}(r)$
and no other correction bit, so $\delta(Z_r)=\mathbf{e}_{\mathrm{fd}(r)}$.
An $X_r$ error becomes $Z$ on the measured qubit under $H$, which commutes
with the $Z$-basis measurement and with every data-qubit stabiliser, so
$\delta(X_r)=0$. Finally $Y_r=iX_rZ_r$ gives
$\delta(Y_r)=\delta(X_r)\oplus\delta(Z_r)=\mathbf{e}_{\mathrm{fd}(r)}$.
Summing over $r$ yields $\delta_v(E)=\bigoplus_{r:\mathrm{fd}(r)=v}[E_r\in\{Y,Z\}]$,
which vanishes for all $v$ precisely under~\eqref{eq:parity_condition}.
\end{proof}

\begin{theorem}[Depolarising noise]
\label{prop:dep}
Under independent depolarising noise
$\mathcal{E}_p(\rho)=(1-p)\rho+(p/3)(X\rho X+Y\rho Y+Z\rho Z)$ on each of the
$k=2|E|$ resource qubits, inserted as specified in
Remark~\ref{rem:noise_model},
\begin{equation}
  \Fstar_\mathrm{dep}(p)
  = \prod_{v\in V}\frac{1+\bigl(1-\tfrac{4p}{3}\bigr)^{\deg_G(v)}}{2}.
  \label{eq:fstar_dep}
\end{equation}
\end{theorem}
\begin{proof}
Each resource qubit feeds exactly one vertex, and vertex $v$ is fed by one
resource qubit per incident edge, i.e.\ by $\deg_G(v)$ of them; since
$\sum_v \deg_G(v)=2|E|=k$, the map $\mathrm{fd}$ \emph{partitions} the
resource qubits. The parity conditions~\eqref{eq:parity_condition} at
distinct vertices therefore involve disjoint sets of independently noisy
qubits and are mutually independent. Under $\mathcal{E}_p$ a given resource
qubit carries a $Y$ or $Z$ error with probability $q=2p/3$. For $n$
independent Bernoulli$(q)$ trials the probability of an even number of
successes is $\tfrac12\bigl(1+(1-2q)^n\bigr)$; taking $n=\deg_G(v)$ and
multiplying over $v$, and using Lemma~\ref{lem:error_parity} to identify
this probability with the fidelity, gives~\eqref{eq:fstar_dep}.
\end{proof}

\begin{theorem}[Phase damping]
\label{prop:pd}
Under independent phase damping with Kraus operators
$\{K_0\!=\!\mathrm{diag}(1,\sqrt{1-p}),\;K_1\!=\!\mathrm{diag}(0,\sqrt{p})\}$
on each of the $k=2|E|$ resource qubits,
\begin{equation}
  \Fstar_\mathrm{pd}(p)
  = \prod_{v\in V}\frac{1+(1-p)^{\deg_G(v)/2}}{2}.
  \label{eq:fstar_pd}
\end{equation}
\end{theorem}
\begin{proof}
Phase damping acts on the Bloch sphere exactly as the dephasing channel
$\rho\mapsto(1-\lambda)\rho+\lambda Z\rho Z$ with
$\lambda=\tfrac12\bigl(1-\sqrt{1-p}\bigr)$, since both scale the off-diagonal
elements by $\sqrt{1-p}$ and leave the populations fixed. It is therefore a
Pauli channel producing a $Z$ error with probability $\lambda$ and no $X$ or
$Y$ component, and Lemma~\ref{lem:error_parity} applies with
$q=\lambda$ and $1-2q=\sqrt{1-p}$. Substituting into the counting argument of
Theorem~\ref{prop:dep} gives~\eqref{eq:fstar_pd}.
\end{proof}

Theorems~\ref{prop:dep} and~\ref{prop:pd} place the channel after the walk
$\CZ$ (Remark~\ref{rem:noise_model}). In a network the dominant noise acts
earlier, on the resource pair while it is generated, transmitted and stored,
and there the fidelity no longer factorises over vertices.

\begin{theorem}[Resource-qubit Pauli noise at either placement]
\label{thm:pre_cz}
Let each resource qubit $r$ undergo an independent Pauli channel
$\rho\mapsto\sum_{\sigma}p^{(r)}_\sigma\,\sigma\rho\sigma$,
$\sigma\in\{I,X,Y,Z\}$, whose probabilities may differ from qubit to qubit.
For $r=r_{e,u}$ write $h(r)=u$ for the party holding it and $\mathrm{fd}(r)$
for the endpoint whose correction bit it feeds
(Lemma~\ref{lem:error_parity}), and $q_r=p^{(r)}_Y+p^{(r)}_Z$.
\begin{enumerate}[label=(\roman*),leftmargin=*]
\item If the channels act \emph{after} the walk $\CZ$, the fidelity of the
  corrected output, averaged over outcomes, is
  \begin{equation}
    \Fstar_{\mathrm{post}}
    =\prod_{v\in V}\frac{1+\prod_{r:\,\mathrm{fd}(r)=v}(1-2q_r)}{2},
    \label{eq:fstar_post_gen}
  \end{equation}
  which reduces to Theorems~\ref{prop:dep} and~\ref{prop:pd} for identical
  channels.
\item If the channels act \emph{before} the walk $\CZ$, on the stored pair,
  \begin{equation}
    \Fstar_{\mathrm{pre}}
    = 2^{-|V|}\sum_{\ell\in\mathbb{F}_2^{V}}\;\prod_{r}\chi_r(\ell),
    \label{eq:fstar_pre}
  \end{equation}
  with
  \begin{align*}
    \chi_r(\ell)={}&p^{(r)}_I+p^{(r)}_X(-1)^{\ell_{h(r)}}\\
    &+p^{(r)}_Y(-1)^{\ell_{h(r)}\oplus\ell_{\mathrm{fd}(r)}}
    +p^{(r)}_Z(-1)^{\ell_{\mathrm{fd}(r)}} .
  \end{align*}
\end{enumerate}
For depolarising noise, with $\lambda=1-4p/3$,
\begin{equation}
  \Fstar_{\mathrm{pre,dep}}(p)
  = 2^{-|V|}\sum_{S\subseteq V}\lambda^{2|E(S)|}
  = 1-2p|E|+O(p^{2}),
  \label{eq:fstar_pre_dep}
\end{equation}
where $E(S)$ is the set of edges with at least one endpoint in $S$; for
$p\le\tfrac34$, $\Fstar_{\mathrm{pre,dep}}\le\Fstar_{\mathrm{dep}}$. For
dephasing channels ($p^{(r)}_X=p^{(r)}_Y=0$ for every $r$) the two placements
give the same fidelity. The theorem concerns independent Pauli channels;
non-Pauli channels such as amplitude damping are not covered by
case~(ii).
\end{theorem}
\begin{proof}
Every operation is Clifford, so, as in Lemma~\ref{lem:error_parity}, a Pauli
error $E$ leaves the corrected state stabilised by $(-1)^{\delta_v(E)}K_v$,
with fidelity $1$ if $\delta(E)=0$ and $0$ otherwise, and $\delta$ is
$\mathbb{F}_2$-linear. A $Z_r$ error commutes through $\CZ(d_u,r)$ and is the
error of Lemma~\ref{lem:error_parity}: $\delta(Z_r)=\mathbf{e}_{\mathrm{fd}(r)}$.
An $X_r$ error is conjugated to $Z_{d_u}X_r$; the $X_r$ factor becomes $Z$
under the pre-measurement Hadamard and leaves the outcome unchanged, while
$Z_{d_u}$ commutes with every later gate and anticommutes with $K_u$ alone, so
$\delta(X_r)=\mathbf{e}_{h(r)}$ and hence
$\delta(Y_r)=\mathbf{e}_{h(r)}\oplus\mathbf{e}_{\mathrm{fd}(r)}$. The fidelity
is $\Pr[\delta=0]=2^{-|V|}\sum_\ell\mathbb{E}\,(-1)^{\ell\cdot\delta}$, and
independence across resource qubits factorises the expectation into the
characters $\chi_r$, giving~\eqref{eq:fstar_pre}. After the $\CZ$ the error
map of Lemma~\ref{lem:error_parity} gives instead
$\chi_r(\ell)=1-2q_r[\ell_{\mathrm{fd}(r)}=1]$; grouping the resource qubits
by the vertex they feed, the sum over $\ell$ factorises over vertices and
gives~\eqref{eq:fstar_post_gen}. For depolarising noise
$\chi_r(\ell)=1$ if $\ell_{h(r)}=\ell_{\mathrm{fd}(r)}=0$ and $\lambda$
otherwise; the two resource qubits of an edge carry the same pair of
endpoints, so with $S=\mathrm{supp}\,\ell$ each edge touching $S$ contributes
$\lambda^{2}$. Averaging $|E(S)|$ over uniform $S$ gives $3|E|/4$, hence the
leading term. The after-$\CZ$ characters are $1$ or $\lambda$ according to
$\ell_{\mathrm{fd}(r)}$ alone; for $0\le\lambda\le1$ each before-$\CZ$
character is no larger and every term is non-negative, so
$\Fstar_{\mathrm{pre,dep}}\le\Fstar_{\mathrm{dep}}$. For dephasing,
$\chi_r$ depends on $\ell_{\mathrm{fd}(r)}$ alone at both placements.
Eqs.~\eqref{eq:fstar_pre} and~\eqref{eq:fstar_pre_dep} were confirmed against
exhaustive Pauli enumeration with full statevector simulation of the
distribution circuit on $P_3$, $P_4$, $K_{1,3}$ and $C_3$, for depolarising,
dephasing and an asymmetric Pauli channel at both placements, and, for
Eqs.~\eqref{eq:fstar_post_gen} and~\eqref{eq:fstar_pre}, with an independently
drawn channel on every resource qubit, agreeing to $<10^{-13}$.
\end{proof}

\begin{remark}[Which placement applies]
\label{rem:placement}
The two placements bracket the protocol. Noise acting after the local $\CZ$
--- gate error and idling before readout --- gives the vertex product of
Theorem~\ref{prop:dep}, whose degree dependence enters only at second order
(Eq.~\eqref{eq:second_order}). Noise acting on the pair before it is used ---
link generation, transmission and storage, the dominant contribution in a
network --- gives the edge-cover sum~\eqref{eq:fstar_pre_dep}, an Ising-type
partition function of the kind obtained by Cuquet and
Calsamiglia~\cite{cuquet2012growth} for their protocols, whose
leading-order loss coefficient is $3/2$ times larger ($2p|E|$ against
$4p|E|/3$); the ratio is a statement about first order in $p$, not about the
fidelity at finite noise. For $P_4$ at $p=0.05$ the two give $0.818$
and $0.739$. Dephasing is placement-independent, a further consequence of
$Z$-transparency (Remark~\ref{rem:cz_props}(i)): only the bit-flip component
of a channel distinguishes the stored pair from the post-gate qubit.
\end{remark}

\begin{remark}[How degree enters the noise budget]
\label{rem:factorisation_scope}
Equations~\eqref{eq:fstar_dep} and~\eqref{eq:fstar_pd} are products over
\emph{vertices}, weighted by degree, rather than flat powers of the
resource-qubit count. When $\Delta(G)\le1$ --- a single edge --- they reduce
to $(1-2p/3)^{k}$ and $((1+\sqrt{1-p})/2)^{k}$ respectively. For
$\Delta(G)\ge2$ they are strictly larger: two errors feeding the same
correction bit $g_v$ cancel in the parity
of~\eqref{eq:parity_condition}, restoring a fidelity that a
per-qubit product would count as lost. The number of such coincidences is
$\sum_v\binom{\deg_G(v)}{2}$, and it is exactly the coefficient of the
leading correction: writing $q$ for the per-qubit $\{Y,Z\}$ probability
($q=2p/3$ for depolarising, $q=\tfrac12(1-\sqrt{1-p})$ for phase damping)
and $x=2q$, taking the logarithm of each vertex factor
of~\eqref{eq:fstar_dep} or~\eqref{eq:fstar_pd} and expanding to second order
gives
\begin{equation}
  \ln\frac{\Fstar}{(1-q)^{2|E|}}
  = \frac{x^{2}}{4}\sum_{v\in V}\binom{\deg_G(v)}{2} + O(x^{3}),
  \label{eq:second_order}
\end{equation}
where $(1-q)^{2|E|}$ is the naive per-resource-qubit product. The gain over
the naive product therefore grows with degree and is unaffected by whether $G$ contains cycles: $P_4$ and $C_3$,
which share $k=6$ and $\Delta=2$, differ only through their degree
sequences. Both formulas were confirmed against exhaustive Pauli and Kraus
enumeration on $K_2$, $P_3$, $P_4$, $K_{1,3}$ and $C_3$, agreeing to
$<10^{-14}$ at every noise strength tested.

That degree rather than particle number governs the noise robustness of
graph states is not itself new. Hein, D\"ur and Briegel~\cite{hein2005entanglement}
showed that the entanglement lifetime of constant-degree graph states is
independent of system size; Cuquet and Calsamiglia~\cite{cuquet2012growth}
obtained decay rates of distributed graph states in terms of the degree
distribution of random networks; and Cavalcanti \emph{et al.}~\cite{cavalcanti2009open}
showed that under Pauli maps the entanglement decay of a graph state can be
read off from its connectivity alone, and the point is reviewed in
Ref.~\cite{aolita2015open}. What Eqs.~\eqref{eq:fstar_dep}--\eqref{eq:fstar_pd}
add is the exact closed form for the resource-qubit channel of the
distribution protocol, where the degree enters through the correction bits
$g_v$ rather than through the state's own stabiliser structure
(Remark~\ref{rem:dataqubit}).

The degree dependence has an operational reading that the hardware data
independently exhibits (Sec.~\ref{sec:hardware}): a high-degree vertex is
penalised twice, once because its party applies one $\CZ$ per incident edge
and accumulates more gate error, and again because its correction bit $g_v$
aggregates more far-side outcomes and so is more easily corrupted. The
stabiliser spectrum of Table~\ref{tab:stabilizer} shows the first effect
directly; Eqs.~\eqref{eq:fstar_dep}--\eqref{eq:fstar_pd} quantify the second.

Three cautions on that reading. First, the gate-count penalty lies outside
the model of Remark~\ref{rem:noise_model}, which takes gates to be noiseless;
only the correction-bit penalty is quantified here. Second, on the regime:
the exact fidelity depends on
the degree sequence at every $p$, but expanding Eq.~\eqref{eq:fstar_dep} for
small $p$ gives $\exp(-4p|E|/3)+O(p^2)$, so the leading exponential decay is
governed by edge count and the degree-dependent corrections enter at second
order, through Eq.~\eqref{eq:second_order}. They are correspondingly small on current hardware: for $K_{1,16}$ the
exact value exceeds the naive per-resource-qubit product by $1.66$ at $p=0.1$
but by only $1.0002$ at the calibrated CZ error of
$1.8\times10^{-3}$. Third, the correction-bit penalty is small, and at fixed
$|E|$ runs the other way once the whole graph is accounted for
(Remark~\ref{rem:noise_scope}):
the star $K_{1,3}$ and the path $P_4$ share $|E|=3$, and
Eq.~\eqref{eq:fstar_dep} gives $0.8189$ and $0.8179$ respectively at
$p=0.05$ --- a difference of $10^{-3}$, with the star \emph{above} the path.
This is Eq.~\eqref{eq:second_order} at work: the two graphs have
$\sum_v\binom{\deg v}{2}=3$ and $2$, giving a predicted ratio $e^{x^2/4}=1.0011$
against the exact $1.0012$.
The measured ordering, $F(\LFour)=0.815(3)$ against
$F(\GHZn{4})=0.776(3)$, is opposite in sign and two orders of magnitude
larger; it is therefore not accounted for by
Eqs.~\eqref{eq:fstar_dep}--\eqref{eq:fstar_pd}, but by the gate-count effect,
which is visible in Table~\ref{tab:stabilizer} and already present in the
direct-preparation baseline of Sec.~\ref{sec:hardware}.
\end{remark}

\begin{remark}[Why $\Fstar_\mathrm{pd} \geq \Fstar_\mathrm{dep}$, and where
amplitude damping sits]
\label{rem:pd_ge_dep}
Both channels enter Lemma~\ref{lem:error_parity} only through the
probability $q$ that a resource qubit acquires a $\{Y,Z\}$ component, and
Eqs.~\eqref{eq:fstar_dep}--\eqref{eq:fstar_pd} are decreasing in $q$ at every
vertex. Depolarising noise gives $q=2p/3$, phase damping
$q=\tfrac12(1-\sqrt{1-p})$, and $\tfrac12(1-\sqrt{1-p})\le 2p/3$ for all
$p\in[0,1]$; hence $\Fstar_\mathrm{pd}(p)\ge\Fstar_\mathrm{dep}(p)$
for every graph and every noise strength. The structural reason is
$Z$-transparency (Remark~\ref{rem:cz_props}(i)): dephasing on a resource
qubit commutes through the walk $\CZ$ and can only disturb a measurement
outcome, never deposit an uncorrectable operator on a data qubit.

The same reduction covers a channel that is not Pauli. Amplitude damping,
$A_0=\mathrm{diag}(1,\sqrt{1-p})$, $A_1=\sqrt{p}\,\ketbra{0}{1}$, decomposes
as $A_0 = \tfrac{1+\sqrt{1-p}}{2}I + \tfrac{1-\sqrt{1-p}}{2}Z$ and
$A_1 = \tfrac{\sqrt p}{2}(X+iY)$. Since $\{I,X\}$ and $\{Y,Z\}$ send the
post-correction state into orthogonal stabiliser eigenspaces
(Lemma~\ref{lem:error_parity}), the cross terms within each Kraus operator
do not contribute to the fidelity, and the effective $\{Y,Z\}$ weight is
$q=\bigl(\tfrac{1-\sqrt{1-p}}{2}\bigr)^{2}+\tfrac{p}{4}$, for which
$1-2q=\sqrt{1-p}$ --- exactly the phase-damping value. Amplitude damping
therefore gives \emph{the same} fidelity as phase damping at the same $p$,
$\Fstar_\mathrm{ad}=\Fstar_\mathrm{pd}$, for every graph, at the placement of
Remark~\ref{rem:noise_model}. This was confirmed
by exhaustive Kraus enumeration on $K_2$, $P_3$, $P_4$, $K_{1,3}$ and $C_3$,
agreeing to $<10^{-14}$ at every $p\in[0,1]$ tested.
\end{remark}

\begin{remark}[Relation to graph-state fidelity under data-qubit noise]
\label{rem:dataqubit}
Closed forms for the fidelity of a graph state under independent Pauli noise
are known for the case in which the channel acts on the $n$ \emph{data}
qubits of an already-prepared $\ket{G}$. Tanizawa \emph{et
al.}~\cite{tanizawa2023} obtain
$F=2^{-n}\sum_{\ell\in\{0,1\}^n}(1-4p/3)^{m(\ell)}$, where $m(\ell)$ is the
number of non-identity Pauli factors in the stabiliser labelled by $\ell$,
and Numajiri \emph{et al.}~\cite{numajiri2025} map that sum to the partition
function of a classical spin system --- a correspondence that Cuquet and
Calsamiglia~\cite{cuquet2012growth} had established earlier for distribution
protocols with noisy channels and gates --- which yields the fidelity of large
regular graph states by transfer-matrix and Monte Carlo methods and locates
a degree-dependent phase transition in it.

Theorems~\ref{prop:dep} and~\ref{prop:pd} concern a different channel and
are not a special case of those results. The noise here acts on the $2|E|$
\emph{resource} qubits of the distribution protocol, between the walk $\CZ$
and the pre-measurement Hadamard (Remark~\ref{rem:noise_model}). The data
qubits are never exposed to it; an error reaches the output only by flipping
a recorded outcome and hence corrupting a correction bit $g_v$
(Lemma~\ref{lem:error_parity}). The two expressions accordingly disagree.
For $K_4$ at $p=0.05$ the data-qubit formula of Ref.~\cite{tanizawa2023}
gives $0.8160$ against $0.6753$ from Eq.~\eqref{eq:fstar_dep}; at $p=0.2$
the values are $0.4268$ and $0.2363$. The two numbers are not a performance
comparison between preparation and distribution --- they describe channels
acting on different registers, and on different numbers of qubits at the
same nominal $p$ --- but they establish that the present formulas are
independent results rather than a re-derivation.
\end{remark}

\begin{remark}[What the closed forms do and do not describe]
\label{rem:noise_scope}
Theorems~\ref{prop:dep} and~\ref{prop:pd} give the fidelity exactly for the
stated model --- independent single-qubit channels on the resource qubits
inserted as in Remark~\ref{rem:noise_model}, with gates and readout
noiseless --- and their role should be read accordingly. Three statements survive contact with real hardware. First,
the \emph{scaling}: because the fidelity is a product over parties, it
decays exponentially in $|V|$ rather than in the resource count, with each
party's factor set by its own degree. A degree-$d$ vertex contributes
$\bigl(1+(1-4p/3)^{d}\bigr)/2$, which decreases with $d$ but saturates at
$\tfrac12$: a hub is the dominant local cost yet never costs more than a
factor of two, so for $p\le\tfrac34$ every factor lies in $[\tfrac12,1]$
and $\Fstar_\mathrm{dep}\ge2^{-|V|}$ irrespective of the degree sequence.
Concentrating edges on a few high-degree vertices is therefore
marginally \emph{cheaper} than spreading them over many low-degree ones at
equal $|E|$, though the effect is small at the sizes tested here and is
swamped on hardware by gate count (Remark~\ref{rem:factorisation_scope}).
To second order this ordering is exact. Since
$\sum_v\binom{\deg v}{2}\le\binom{|E|}{2}$, with equality if and only if every
two edges share a vertex --- the star $K_{1,|E|}$, or $C_3$ when $|E|=3$ ---
these are, by Eq.~\eqref{eq:second_order}, the most noise-robust targets at
fixed $|E|$ for sufficiently small $p$. $C_3$ and $K_{1,3}$ agree through
third order; exactly, with $x=2q$ as in Eq.~\eqref{eq:second_order},
$\Fstar_{C_3}-\Fstar_{K_{1,3}}=x^4(x^2-3x+3)/16>0$, so $C_3$ is the more
robust of the two at every noise strength. This ranks target graphs under the
channel of Remark~\ref{rem:noise_model}; it is not a prediction of device
fidelity, and the ordering measured on \texttt{ibm\_kingston} is the reverse
(Remark~\ref{rem:factorisation_scope}).
Its size grows with degree: for the star $K_{1,n}$ at $p=0.1$ the exact
fidelity exceeds $(1-2p/3)^{2n}$ --- the value that would obtain if every
one of the $2n$ resource qubits failed independently, with no parity
cancellation --- by a factor of $1.66$ at $n=16$.
Second, the
\emph{channel ordering}: dephasing on resource qubits commutes through
$\CZ$ (Remark~\ref{rem:cz_props}(i)) and is absorbed by the $Z$-only
correction, so phase damping is the least destructive channel --- a
structural property of the protocol that holds at any noise strength.
Third, and by design \emph{not} claimed: the closed forms are not
predictors of device fidelity. The hardware analysis
(Sec.~\ref{sec:hardware}) shows that a joint fit of the measured
generators and fidelity to independent local Pauli channels is infeasible
(residual $1.5\times10^{-2}$), i.e.\ the device noise is correlated or
coherent beyond any independent single-qubit model; for this reason no
effective single-qubit rate is extracted from the data. The model
nevertheless serves as a ruler: simulating the native $\LFour$ circuit
under independent depolarising and readout noise at the calibrated device
rates predicts $F=0.889$, against the measured $0.815(3)$ ---
the $0.074$ shortfall quantifies the contribution of noise
structure beyond the independent model.
\end{remark}

\begin{corollary}[The corrected output is graph-diagonal]
\label{cor:graph_diagonal}
Under the noise model of Remark~\ref{rem:noise_model}, the corrected data
state averaged over measurement outcomes is
\begin{equation}
  \rho(p) = \sum_{\delta\in\{0,1\}^{V}} P_p(\delta)\,
  \Bigl(\prod_v Z_v^{\delta_v}\Bigr)\ketbra{G}{G}
  \Bigl(\prod_v Z_v^{\delta_v}\Bigr),
  \label{eq:graph_diagonal}
\end{equation}
a graph-diagonal state whose syndrome bits are independent across vertices,
\begin{equation}
  \Pr[\delta_v=0] = \frac{1+\bigl(1-\tfrac{4p}{3}\bigr)^{\deg_G(v)}}{2},
  \label{eq:syndrome}
\end{equation}
and correspondingly for phase damping with $1-4p/3$ replaced by $(1-p)^{1/2}$.
\end{corollary}
\begin{proof}
By Lemma~\ref{lem:error_parity} the post-correction state for a Pauli error
$E$ is stabilised by $(-1)^{\delta_v(E)}K_v$, i.e.\ it is
$\prod_v Z_v^{\delta_v}\ket{G}$, and $\delta_v$ is the parity of the
$\{Y,Z\}$-indicators of the $\deg_G(v)$ resource qubits with
$\mathrm{fd}(r)=v$. Since $\mathrm{fd}$ partitions the resource qubits
(proof of Theorem~\ref{prop:dep}), distinct $\delta_v$ are parities of
disjoint sets of independently noisy qubits, giving independence and
Eq.~\eqref{eq:syndrome}. Eq.~\eqref{eq:fstar_dep} is the $\delta=0$
component. Independently of Lemma~\ref{lem:error_parity},
Eq.~\eqref{eq:graph_diagonal} was confirmed by full density-matrix
simulation of the distribution circuit --- resource-state preparation, walk
$\CZ$, channel, Hadamard, measurement of every resource qubit, and
outcome-weighted correction --- for $P_3$, $P_4$ and $K_{1,3}$, agreeing
with the closed form to $<10^{-15}$ at every $p$ tested.
\end{proof}

\begin{remark}[State functionals of the output, and a thermodynamic reading]
\label{rem:ergotropy}
Corollary~\ref{cor:graph_diagonal} localises the entire $p$-dependence of
the output in the classical syndrome distribution: the eigenvectors of
$\rho(p)$ are the $2^{|V|}$ signed graph states, fixed once $G$ is fixed,
and only their weights move. Any functional of $\rho(p)$ that sees the state
through its spectrum alone is therefore determined by the same per-vertex
probabilities that give Theorem~\ref{prop:dep}, hence by the degree sequence
of $G$ and not by its topology; the two non-isomorphic connected graphs on
six vertices with degree sequence $(1,1,1,2,2,3)$ give identical values of
every such functional, as checked numerically.

One instance is worth recording, since graph states are a natural candidate
resource for distributed quantum batteries. The ergotropic
gap~\cite{mukherjee2016presence,alimuddin2019bound,puliyil2022thermodynamic}
is the difference between globally and locally extractable work, and is an
LOCC monotone for pure states. For a non-interacting $H=\sum_v h_v$ with
equal level spacing, every signed graph state has maximally mixed
single-qubit marginals, so the local ergotropy vanishes identically and the
gap of each eigenvector of Eq.~\eqref{eq:graph_diagonal} equals $|V|/2$,
independently of $\delta$, of $p$, and of $G$. The convex-roof extension of
the gap is thus constant along the channel, while the ergotropy of $\rho(p)$
evaluated directly is spectrum-only in the sense above and so carries no
information beyond Theorem~\ref{prop:dep}. Neither quantity separates the
distributed graph state from any other output with maximally mixed
marginals.

For an interacting Hamiltonian the spectrum-only argument fails, but for a
Pauli Hamiltonian the energy remains fixed by which of its terms lie in
$\mathrm{Stab}(\ket{G})$, since $\bra{G}P\ket{G}=\pm1$ for
$P\in\mathrm{Stab}(\ket{G})$ and $0$ otherwise. For weight-two terms this is
purely combinatorial: $X_uX_v$ belongs to $\mathrm{Stab}(\ket{G})$ exactly
when $u,v$ are non-adjacent with $N(u)=N(v)$, $Y_uY_v$ exactly when they are
adjacent with $N(u)\setminus\{v\}=N(v)\setminus\{u\}$, and no product of $Z$
operators alone ever does. This condition is unrelated to the native
embeddability of the canonical interaction graph
(Sec.~\ref{sec:embeddability}): all four combinations of (embeddable,
energetically active) occur among the graphs considered here, with
$K_{1,3}$ and $C_4$ in one corner and $P_4$ and $K_4$ in the two others.
\end{remark}

\begin{figure}[t]
  \centering
  \includegraphics[width=\columnwidth]{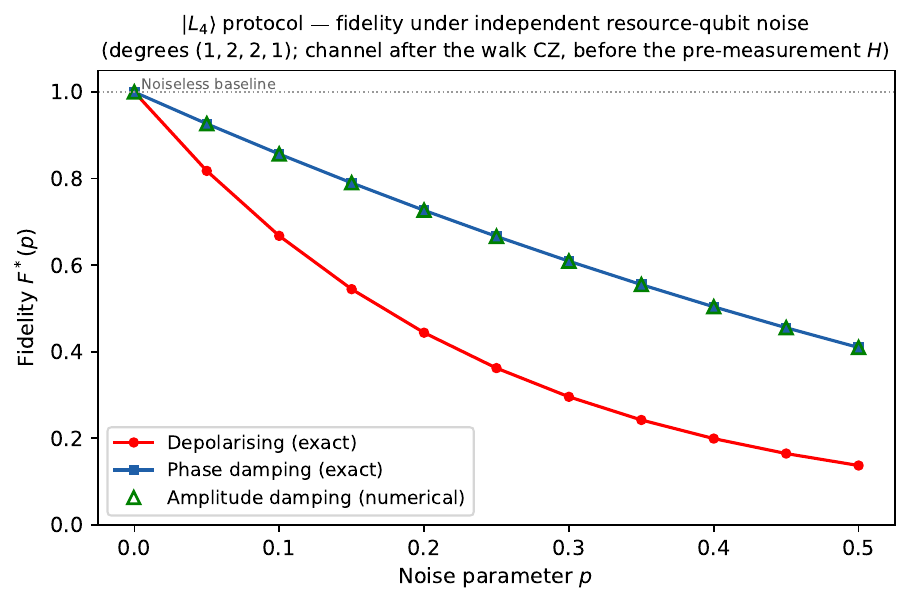}
  \caption{Fidelity $\Fstar(p)$ for the $\LFour$ protocol under
    independent noise on each resource qubit, inserted after the walk $\CZ$
    and before the pre-measurement Hadamard (Remark~\ref{rem:noise_model}):
    depolarising (red) and phase damping (blue) from
    Theorems~\ref{prop:dep} and~\ref{prop:pd}, amplitude damping
    (green, numerical). The analytical curves are exact for the stated
    model; $\LFour$ has degree sequence $(1,2,2,1)$. The amplitude-damping
    points lie on the phase-damping curve: the two channels give identical
    fidelities in this model (Remark~\ref{rem:pd_ge_dep}), so the green
    markers are a numerical check of the analytical blue curve rather than a
    third channel. Phase damping is least destructive due to
    $Z$-transparency of $\CZ$.}
  \label{fig:fstar_decay}
\end{figure}

\begin{figure}[t]
  \centering
  \includegraphics[width=\columnwidth]{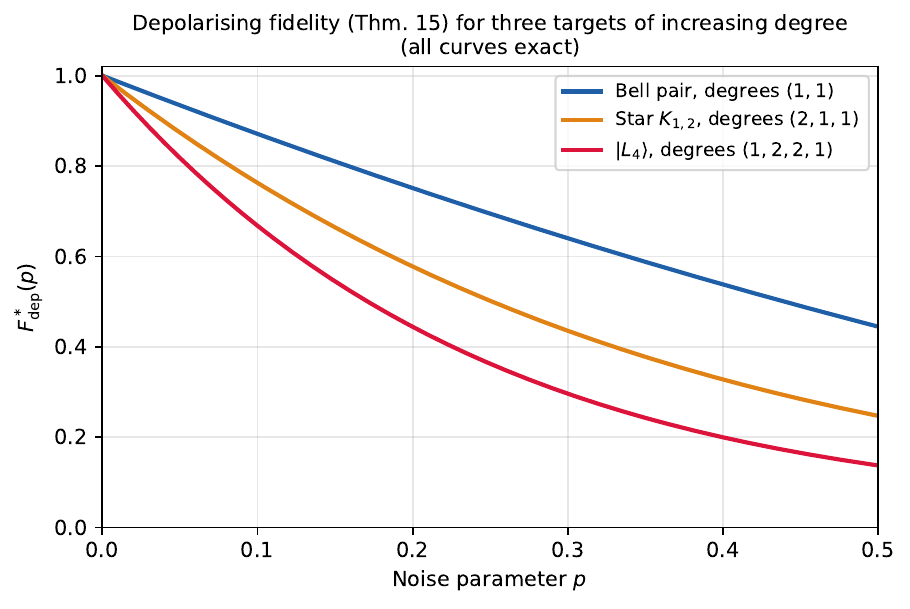}
  \caption{Depolarising fidelity $\Fstar_\mathrm{dep}(p)$ of
    Theorem~\ref{prop:dep} for three targets with increasing degree:
    a Bell pair (degrees $1,1$), a three-party star (degrees $2,1,1$), and
    $\LFour$ (degrees $1,2,2,1$). All are exact.
    Comparing the star with $\LFour$ isolates the role of the degree
    sequence at fixed resource count.}
  \label{fig:comparison}
\end{figure}

\subsection{Hardware Validation on IBM Quantum}
\label{sec:hardware}

Executed on \texttt{ibm\_kingston} (IBM Heron~r2, 156 qubits,
CZ-native)~\cite{IBMQuantum} in a single batch with $N=4096$ shots per
measurement setting. Layouts were pinned to native (SWAP-free) qubit
subgraphs under a quality policy applied \emph{per embedding class}, because
no embedding meeting a single uniform cut exists for all three targets.
$\LFour$ was placed on a chain with every pinned qubit at
$T_2\ge147\,\mu\mathrm{s}$ and readout error below $2\%$. For $\GHZn{4}$ the
device offers no degree-3 spider meeting that cut, and the layout was chosen
at a relaxed threshold ($T_2>90\,\mu\mathrm{s}$, readout error $<3\%$). For
$\ket{C_4}$ none of the 21 hexagonal faces is clean; the face used has
minimum $T_2=96\,\mu\mathrm{s}$ and maximum readout error $3.8\%$, with the
weakest-readout qubit placed on a resource rather than a data position. The
relaxations are on $T_2$ and readout error only. A submission-time drift
check confirmed that every pinned qubit met its class threshold at
execution.

\paragraph{Fidelity measure: direct fidelity estimation.}
For a stabiliser target state $\ket{G}$ on $n$ data qubits, the fidelity
of the corrected output $\rho$ is
\begin{equation}
  F = \bra{G}\rho\ket{G}
    = \frac{1}{2^{n}}\sum_{S\in\mathcal{S}(G)}\langle S\rangle_\rho,
  \label{eq:dfe}
\end{equation}
where $\mathcal{S}(G)$ is the $2^{n}$-element stabiliser group of $\ket{G}$
and $\langle S\rangle_\rho$ its measured expectation. This is exact (not a
bound) and requires no full tomography~\cite{flammia2011direct,dasilva2011practical}: the $2^{n}-1$ non-identity group
elements are grouped into qubitwise-commuting measurement settings --- nine
per protocol here --- from which every $\langle S\rangle$ is reconstructed.
Nine is the exact minimum, and not by a counting argument: each of the three
stabiliser groups contains nine elements that are pairwise \emph{not}
co-measurable in any single product basis, which forces at least nine
settings, and the nine settings used achieve it. For $\GHZn{4}$ minimality is
immediate --- nine of the fifteen elements admit exactly one admissible
product basis each. The same minimum holds for every connected graph state on
four vertices. The $Z$-only
corrections of Theorem~\ref{thm:universal_z} do not alter
computational-basis probabilities, so each $g_v$ is applied in classical
post-processing shot-by-shot; no mid-circuit feed-forward is required.

\paragraph{Why the classical (distribution) fidelity is uninformative here.}
A natural and inexpensive alternative to~\eqref{eq:dfe} is a classical
(Bhattacharyya) statistic read off the protocol's own shots, needing no
measurement settings beyond the protocol itself. Conditioned on each
resource-measurement outcome $r$, form the Bhattacharyya coefficient between
the empirical distribution $\{p_{d|r}\}$ over the $2^{|V|}$ data-qubit
outcomes and the uniform reference $\{q_d=2^{-|V|}\}$, square it, and average
over classes weighted by class probability:
\begin{equation}
  F^*_\mathrm{cl} \;=\; \sum_r \frac{m_r}{N}
    \Bigl(\sum_d \sqrt{p_{d|r}\,q_d}\Bigr)^{2},
  \label{eq:fcl}
\end{equation}
with $m_r$ the number of shots in class $r$ and $N$ the total. I record here that this statistic is structurally noise-blind for
graph states, so that it is not mistaken for a fidelity diagnostic: the ideal
distribution is uniform, and depolarising noise also drives the measured
distribution toward uniform, so $F^*_\mathrm{cl}$ compares uniform to
uniform. At $N=4096$ shots the noiseless finite-shot floor of this estimator
is $0.929$, $0.623$, and $0.119$ for $k=6,8,12$ resource qubits
(Fig.~\ref{fig:fcl_floor}). An earlier version of this manuscript reported
measured values of $0.929$ and $0.924$ at $k=6$ ($\GHZn{4}$, $\LFour$;
Fig.~\ref{fig:fcl_floor} plots their mean, $0.927$),
$0.616$ at $k=8$ ($\ket{C_4}$), and $0.120$ at $k=12$ ($\ket{K_4}$): each
coincides with the floor for its own $k$ to within shot noise and therefore
carries no information about device error. The claim of topology-independent
fidelity drawn from those values does not survive this analysis and is
withdrawn; $F^*_\mathrm{cl}$ is replaced throughout by the direct fidelity
estimate~\eqref{eq:dfe}.

\begin{figure}[t]
  \centering
  \includegraphics[width=0.72\columnwidth]{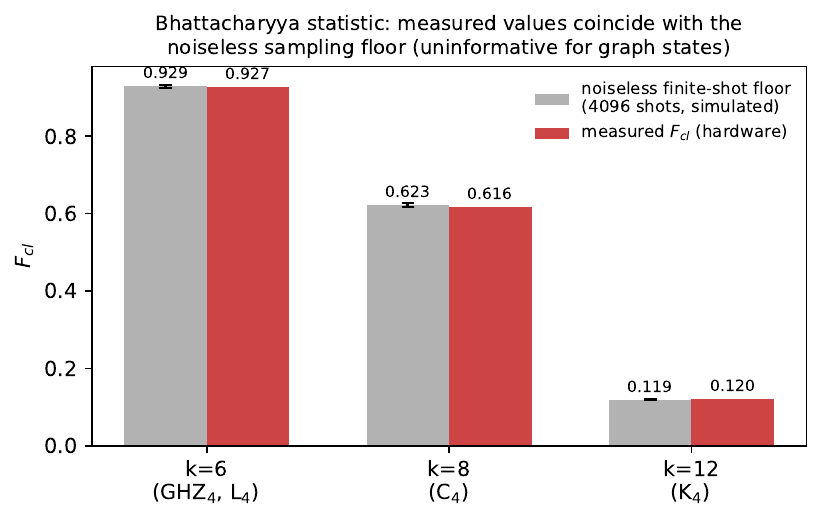}
  \caption{The Bhattacharyya statistic $F^*_\mathrm{cl}$ is uninformative
    for graph states: the measured hardware values (red) coincide with the
    noiseless finite-shot sampling floor (grey, $N=4096$ simulated) for
    every resource-qubit count $k$, because both the ideal graph-state
    distribution and the depolarised distribution are uniform. Direct
    fidelity estimation~\eqref{eq:dfe} is used instead.}
  \label{fig:fcl_floor}
\end{figure}

\paragraph{Results.}
Table~\ref{tab:hardware} reports the direct fidelity estimate for the three
natively embeddable protocols; Table~\ref{tab:stabilizer} gives the
individual stabiliser-generator expectations extracted from the same batch.
The path state $\LFour$ reaches $F=0.815(3)$, the cycle $\ket{C_4}$
$F=0.733(3)$, and the star $\GHZn{4}$ $F=0.776(3)$; the quoted uncertainties
are shot-noise standard errors propagated per generator with the exact
within-setting covariance. All three exceed the generator-based lower bound
$F\ge 1-\sum_v(1-\langle K_v\rangle)/2$ (0.765, 0.653, 0.701 respectively),
which follows from $\ketbra{G}{G}=\prod_v P_v$ with $P_v=(I+K_v)/2$ and the
operator inequality $I-\prod_v P_v\le\sum_v(I-P_v)$ for commuting projectors
(cf.\ the stabiliser witnesses of Refs.~\cite{toth2005detecting,toth2005entanglement}),
an internal consistency check on the reconstruction. The canonical $K_4$
protocol is absent from the table by necessity, not omission: its interaction
graph cannot be embedded natively on heavy-hex
(Theorem~\ref{thm:embeddability}). The state $\ket{K_4}$ is nevertheless
present in the table, under the name $\GHZn{4}$, by the LC equivalence
discussed in Sec.~\ref{sec:embeddability}.

\paragraph{Direct-preparation baseline: the cost of distribution.}
To make the distributed fidelities interpretable I measured, on the same
device and under a policy of the same form --- a live layout solve against
$T_2$, readout and CZ-error thresholds --- the \emph{co-located ceiling}:
each target state prepared directly on four adjacent qubits with three
local CZ gates and certified by the identical DFE procedure
($N=8192$ shots per setting). This baseline is not a competing protocol
--- it presumes all four parties share one device patch, which is
precisely the assumption distribution removes --- but it bounds what the
hardware can do for these states with no distribution machinery at all.
The measured ceilings are $F_{\rm dir}=0.938(1)$ for $\LFour$ and
$0.885(2)$ for $\GHZn{4}$ (Table~\ref{tab:hardware}), so the distribution
overhead is $0.124(3)$ and $0.109(4)$ respectively (computed from the
unrounded fidelities) --- a roughly uniform
cost of $\approx 0.11$--$0.12$ for replacing local two-qubit gates between
co-located parties by shared $\GKtwo$ resources, local operations, and
classical communication. Two further observations. First, the direct-prep
layouts use four \emph{adjacent} qubits while the distributed protocol's
data qubits are deliberately non-adjacent (resource qubits sit between
them); this is not a confound but the object of comparison --- co-location
versus separation is exactly what is being priced. Second, the hub
suppression appears already in the baseline: direct $\GHZn{4}$
($0.885$) sits well below direct $\LFour$ ($0.938$) despite identical
gate counts, so part of the distributed star's lower fidelity is intrinsic
to preparing a degree-$3$ hub on this hardware rather than a cost of
distribution. No baseline is reported for $\ket{C_4}$ because none exists
natively: direct preparation requires a $4$-cycle in the coupling graph,
which heavy-hex lacks (Remark~\ref{rem:direct_embed}) --- for the box
cluster, the distributed protocol is not merely competitive with native
direct preparation but is, of the two, the only one that executes natively
at all.

\paragraph{Comparison with prior walk-based hardware.}
The closest experimental point of reference is Chen \emph{et
al.}~\cite{chen2025entanglement}, who implemented walk-based distribution
modules on a ten-qubit superconducting chain and certified the outputs by
full state tomography ($10^4$ shots). They report average fidelities of
$0.798(20)$ for a two-party Bell state distributed from two Bell states,
$0.543(36)$ for a three-qubit GHZ state from two GHZ states, $0.544(11)$
for a three-party GHZ state from three Bell states, and $0.470(7)$ for a
three-party GHZ state from three GHZ states. The present results ---
$0.815(3)$, $0.776(3)$ and $0.733(3)$ for \emph{four}-party graph states,
including the non-GHZ topologies $\LFour$ and $\ket{C_4}$ --- are obtained
for strictly larger targets. The comparison should be read with two caveats:
the devices differ by generation (their reported $T_2^{*}$ values are
$1.2$--$2.7\,\mu\mathrm{s}$, against $T_2\approx148\,\mu\mathrm{s}$ here,
Table~\ref{tab:calibration}), so a substantial part of the difference is
hardware rather than protocol; and the certification methods differ, full
tomography versus direct fidelity estimation, although both report state
fidelity to the intended target. What the comparison does establish is that
distributing four-party graph states beyond the GHZ class is within reach of
current superconducting hardware at fidelities comparable to, or above, those
previously reported for smaller GHZ-class targets by walk-based methods.

\begin{table}[t]
\caption{Direct fidelity estimation on \texttt{ibm\_kingston}
(single batch, native SWAP-free layouts, $N=4096$ shots per setting,
nine settings per protocol). $F$ is the exact stabiliser
fidelity~\eqref{eq:dfe} of the \emph{distributed} state; the parenthesised
digit is the shot-noise standard error on the last place. ``CZ'' is the
transpiled two-qubit gate count per setting circuit (native, no SWAP
inflation). $F_{\rm dir}$ is the measured co-located ceiling: the same
state prepared \emph{directly} on four adjacent qubits (three local CZ
gates, same device, live layout solve under thresholds of the same form,
$N=8192$).
$^{\dagger}$Direct native preparation of $\ket{C_4}$ is impossible on
heavy-hex: it requires a $4$-cycle in the coupling graph, and heavy-hex
(girth $12$) contains none (Remark~\ref{rem:direct_embed}).
$^{\ddagger}$The last row refers to the canonical $K_4$ interaction graph
$I(K_4)$, not to the state $\ket{K_4}$: since $K_4\sim_{\rm LC}K_{1,3}$, the
state is obtainable on the star layout of row two
(Sec.~\ref{sec:embeddability}).
The embedding column records how much freedom the interaction graph leaves in
choosing physical qubits: \emph{flexible} (any sufficiently long clean chain),
\emph{degree-3 constrained} (a spider about a degree-3 vertex, of which none
meets the $\LFour$ threshold on this device), and \emph{rigid} (one entire
hexagonal face, none of which is clean). Sec.~\ref{sec:embeddability} argues
that this freedom, not the topology as such, is what the measured fidelities
track.}
\label{tab:hardware}
\begin{ruledtabular}
\begin{tabular}{lrrrl}
Protocol & CZ & $F$ & $F_{\rm dir}$ & Embedding \\
\hline
$\LFour$      &  9 & $0.815(3)$ & $0.938(1)$ & path (flexible) \\
$\GHZn{4}$    &  9 & $0.776(3)$ & $0.885(2)$ & star (degree 3) \\
$\ket{C_4}$   & 12 & $0.733(3)$ & --\,$^{\dagger}$ & 12-cycle (rigid) \\
$I(K_4)^{\ddagger}$ & -- & --    & -- & none (Thm.~\ref{thm:embeddability}) \\
\end{tabular}
\end{ruledtabular}
\end{table}

\paragraph{Relation to dynamic-circuit gate teleportation.}
Long-range entangling gates built from measurement and feed-forward have
been demonstrated on IBM heavy-hex devices at much larger scale, including
CNOT gate teleportation across up to 101 qubits~\cite{baumer2024efficient}.
The present experiment uses the same class of primitive for a different task
--- certifying a distributed multi-party graph state --- and, because its
corrections are $Z$-only, applies them in classical post-processing rather
than by mid-circuit feed-forward.

\paragraph{Stabiliser structure: the degree imprint.}
The per-generator expectations (Table~\ref{tab:stabilizer},
Fig.~\ref{fig:dfe_generators}) reveal that the target graph's degree
structure imprints directly on the stabiliser spectrum. For $\GHZn{4}$ the
hub generator $K_{\rm hub}=X_{\rm hub}Z_{B}Z_{C}Z_{D}$ measures
$0.714$, far below the three leaf generators ($0.89$--$0.90$); the hub party
applies three CZ gates (one per adjacent edge) while each leaf applies one,
so the weight-four hub stabiliser accumulates the most error. For $\LFour$
the two interior (weight-three) generators $K_B,K_C$ ($0.837$, $0.848$)
sit below the two endpoint (weight-two) generators $K_A,K_D$ ($0.939$,
$0.906$), the milder analogue of the same effect on a degree-two path.
The asymmetry is a property of high-degree nodes rather than of
distribution: as recorded above it is already present in the co-located
baseline, which involves no resource qubits, no measurement and no
correction at all. A joint fit of the
generators and the fidelity to independent local Pauli channels is
infeasible (residual $1.5\times10^{-2}$), indicating correlated or coherent
error beyond an independent single-qubit model; I therefore report the
measured spectrum without imposing such a model.

\begin{table}[t]
\caption{Stabiliser-generator expectations on \texttt{ibm\_kingston},
extracted from the same DFE batch. $\langle K_v\rangle$ is the corrected
expectation (Sec.~\ref{sec:universal_correction}); $F\ge$ is the fidelity
lower bound $1-\sum_v(1-\langle K_v\rangle)/2$, shown alongside the exact
$F$ from Table~\ref{tab:hardware}. For $\GHZn{4}$, $K_0$ is the hub.}
\label{tab:stabilizer}
\begin{ruledtabular}
\begin{tabular}{lrrrrrr}
Protocol & $\langle K_0\rangle$ & $\langle K_1\rangle$ &
$\langle K_2\rangle$ & $\langle K_3\rangle$ & $F\geq$ & $F$ \\
\hline
$\LFour$      & 0.939 & 0.837 & 0.848 & 0.906 & 0.765 & 0.815 \\
$\GHZn{4}$    & 0.714 & 0.904 & 0.894 & 0.890 & 0.701 & 0.776 \\
$\ket{C_4}$   & 0.811 & 0.898 & 0.802 & 0.795 & 0.653 & 0.733 \\
\end{tabular}
\end{ruledtabular}
\end{table}

\begin{figure}[t]
  \centering
  \includegraphics[width=\columnwidth]{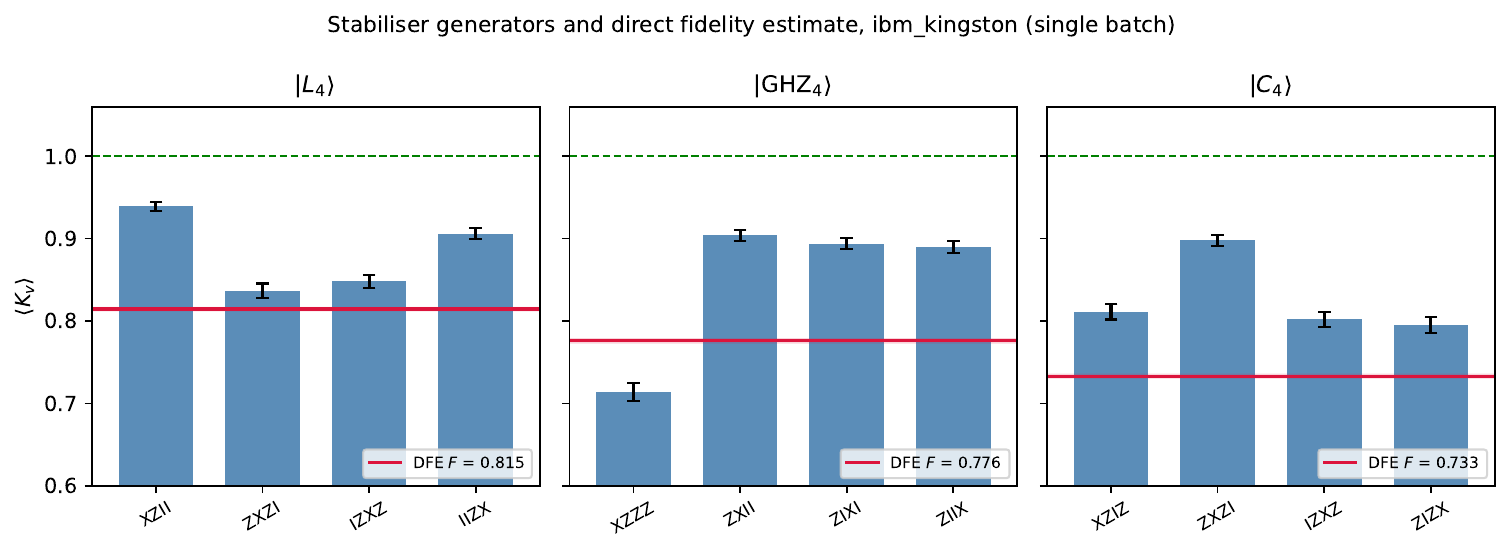}
  \caption{Stabiliser generators $\langle K_v\rangle$ (bars, with shot-noise
    error bars) and direct fidelity estimate (red line, band $=\pm\sigma$)
    for the three natively embeddable protocols on \texttt{ibm\_kingston},
    single batch. The hub generator of $\GHZn{4}$ (leftmost, $0.714$) and
    the interior generators of $\LFour$ lie well below the high-weight and
    endpoint generators respectively, exhibiting the degree imprint
    discussed in the text. The green dashed line marks the ideal value
    $+1$.}
  \label{fig:dfe_generators}
\end{figure}

\paragraph{Native embedding and the coherent-correction alternative.}
The $Z$-only correction (Theorem~\ref{thm:universal_z}) is applied in
classical post-processing, so the executed circuit is exactly the
distribution circuit of Definition~\ref{def:general_protocol}: prepare
$\GKtwo$ per edge, one CZ per data--resource interaction, measure. Its
interaction graph has maximum degree $2$ ($\LFour$), $3$ ($\GHZn{4}$), and
$2$ ($\ket{C_4}$), all within the heavy-hex maximum of $3$, and transpiles
to $9$, $9$, and $12$ CZ gates with no SWAP insertion. A tempting
alternative --- deferring the correction into coherent CZ gates rather than
classical post-processing --- inflates the data-qubit degree to $4$
($\LFour$) and $6$ ($\GHZn{4}$), exceeding the heavy-hex maximum and forcing
SWAP routing that roughly doubles the two-qubit gate count; in a preliminary
run this coherent formulation reached only $F\approx0.68$--$0.69$, some
$0.13$ below the native classical-correction result. The classical
feed-forward that Theorem~\ref{thm:universal_z} enables is therefore not
merely operationally convenient but hardware-native.

\paragraph{Device calibration and stationarity.}
Table~\ref{tab:calibration} lists the \texttt{ibm\_kingston} calibration at
the time of execution. Coherence on this device is non-stationary over the
batch timescale: one pinned qubit's $T_2$ collapsed from $214\,\mu\mathrm{s}$
to $22\,\mu\mathrm{s}$ between the layout-selection and a subsequent
calibration, a factor-of-ten swing characteristic of two-level-system
defects. This motivates the single-batch execution used here --- all three
protocols share the submission window and the same calibration epoch, so the
cross-protocol comparison is not confounded by drift --- and it is the
mechanism by which layouts of nominally equivalent calibrated quality can
differ in measured fidelity. The same
qubit's $T_2$ subsequently recovered to $141\,\mu\mathrm{s}$, confirming the
episodic character of the defect. The quoted uncertainties in
Tables~\ref{tab:hardware} and~\ref{tab:stabilizer} are consequently
shot-noise only: the single-batch design bounds statistical error but not
run-to-run variance under recalibration, which this dataset does not
constrain and which the observed layout and epoch sensitivity suggests is
of order a few times $10^{-2}$.

\begin{table}[t]
\caption{Device calibration data, \texttt{ibm\_kingston}, at execution
(2026-07-05).}
\label{tab:calibration}
\begin{ruledtabular}
\begin{tabular}{lr}
Parameter & Median value \\
\hline
CZ gate error & $1.81\times10^{-3}$ \\
$T_1$ & $248\,\mu\mathrm{s}$ \\
$T_2$ & $148\,\mu\mathrm{s}$ \\
Readout assignment error & $8.1\times10^{-3}$ \\
Mid-circuit readout error & $8.4\times10^{-3}$ \\
Readout duration & $2180\,$ns \\
\end{tabular}
\end{ruledtabular}
\end{table}

\subsection{Native Embeddability on Heavy-Hex Hardware}
\label{sec:embeddability}

The three protocols above occupy qualitatively different positions with
respect to the device connectivity, and this is not incidental: it follows
from a girth constraint that I now make precise.

The PQW distribution circuit places two resource qubits on every edge
(one per endpoint), so the physical \emph{interaction graph} $I(G)$ ---
the graph whose vertices are all data and resource qubits and whose edges
are the CZ interactions --- is obtained from the target graph $G$ by
subdividing each edge twice. Subdivision multiplies every cycle length by a
fixed factor: an $\ell$-cycle in $G$ becomes a $3\ell$-cycle in $I(G)$.
Consequently the girth satisfies $\mathrm{girth}(I(G)) = 3\,\mathrm{girth}(G)$
(with $\mathrm{girth}=\infty$ for acyclic $G$), and the maximum degree
satisfies $\Delta(I(G)) = \max(\Delta(G),2)$, since each data qubit $v$
retains degree $\deg_G(v)$ while each resource qubit has degree $2$.

\begin{theorem}[Necessary conditions for native embedding]
\label{thm:embeddability}
Let $H$ be a hardware coupling graph of maximum degree $3$ and girth
$\gamma$ (heavy-hex has $\Delta(H)=3$, $\gamma=12$). If the PQW interaction
graph $I(G)$ embeds as a subgraph of $H$ without SWAP insertion, then the
target graph $G$ satisfies
\begin{enumerate}[label=(\roman*),leftmargin=*]
  \item $\Delta(G)\le 3$, and
  \item $\mathrm{girth}(G)\ge \gamma/3$; equivalently, $G$ contains no cycle
        shorter than $\gamma/3$. For heavy-hex this forbids triangles
        ($G$ must be triangle-free).
\end{enumerate}
\end{theorem}
\begin{proof}
A subgraph embedding preserves both degree and girth as inequalities.
For (i): a data qubit $v$ has degree $\deg_G(v)$ in $I(G)$, which cannot
exceed $\Delta(H)=3$. For (ii): a shortest cycle of $G$, of length
$\mathrm{girth}(G)$, becomes a cycle of length $3\,\mathrm{girth}(G)$ in
$I(G)$; a subgraph of $H$ cannot contain a cycle shorter than $\gamma$, so
$3\,\mathrm{girth}(G)\ge\gamma$. For heavy-hex, $\gamma/3=4$, so any triangle
($3$-cycle) in $G$ would require a $9$-cycle in $I(G)$, impossible in a
girth-$12$ graph.
\end{proof}

\begin{corollary}[$I(K_4)$ is not natively embeddable on heavy-hex]
\label{cor:k4}
$K_4$ contains triangles, so its interaction graph contains $9$-cycles,
which cannot occur as subgraphs of a girth-$12$ coupling graph. The
canonical $K_4$ distribution circuit therefore admits no SWAP-free embedding
on any heavy-hex device, independent of qubit selection or calibration.
\end{corollary}

These conditions are \emph{necessary}, not sufficient: they exclude
the canonical $K_4$ circuit but do not by themselves guarantee that an
arbitrary triangle-free, degree-$\le 3$ graph embeds, since embedding also
requires a suitable matching of $I(G)$ into the specific lattice.

They also constrain the \emph{circuit}, not the \emph{state}. Local Cliffords
are free in this operational model, so both the resource cost and the
embeddability of a target are properties of its local-Clifford orbit rather
than of the particular graph written
down~\cite{schroeder2024deterministic,van2004graphical,cabello2011optimal,sharma2026minimising}. $\ket{K_4}$ is the
sharp case: one local complementation at the hub of $K_{1,3}$ completes its
independent neighbourhood into a triangle, so $K_4\sim_{\rm LC}K_{1,3}$ and
$\ket{K_4}$ is obtainable from \emph{three} resource pairs on the natively
embeddable star layout followed by one layer of single-qubit Cliffords ---
indeed it is the state measured at $F=0.776(3)$ in
Table~\ref{tab:hardware} under the name $\GHZn{4}$.
Corollary~\ref{cor:k4} is therefore a statement about the canonical
interaction graph $I(K_4)$, not an impossibility for the state. By the same
token the $|E|$ pairs consumed by the protocol are a sufficiency statement for
the representative one chooses to distribute, not a state-level optimum: that
optimum is $\min\{|E'| : G'\in\mathrm{LC}(G)\}$, which is strictly smaller for
$C_3$, $C_4$, the diamond and $K_4$ among others. Trees are where the two
coincide, which is why the cut argument of Sec.~\ref{sec:intro} is tight there
and only there. The three tested protocols
illustrate the resulting hierarchy. The path $\LFour$ (acyclic, $\Delta=2$)
embeds flexibly: its interaction graph is a chain that fits any sufficiently
long run of qubits, leaving freedom to select high-coherence qubits.
The star $\GHZn{4}$ ($\Delta=3$, acyclic) embeds as a depth-one spider about
a degree-$3$ vertex. The cycle $\ket{C_4}$ (girth $4$, $\Delta=2$) is the
boundary case: its interaction graph is a $12$-cycle, exactly matching the
heavy-hex girth, so it embeds \emph{rigidly} --- only on the device's
hexagonal faces, of which \texttt{ibm\_kingston} has $21$, with no freedom
to avoid a blemished qubit within a chosen face. $I(K_4)$ lies past the
boundary and does not embed at all. The measured fidelities
(Table~\ref{tab:hardware}) descend along this hierarchy, from the flexible
path to the rigid cycle, with the non-embeddable complete graph absent by
Corollary~\ref{cor:k4}.

The descent should not be read as a correlation between rigidity and fidelity
with layout quality as a competing explanation, because layout quality is the
\emph{mechanism} by which rigidity is paid for. Flexibility is precisely the
freedom to choose which physical qubits carry the protocol: $\LFour$
subdivides to a chain, so it may be placed on the best run of qubits the
device offers, and it was ($T_2 \ge 147\,\mu\mathrm{s}$, readout $<2\%$;
Sec.~\ref{sec:hardware}). A rigid embedding forfeits that freedom by
definition --- $\ket{C_4}$ must occupy one of the 21 hexagonal faces in its
entirety, and on this device none of the 21 is clean, so the best available
placement still carries $T_2$ down to $96\,\mu\mathrm{s}$ and readout error up
to $3.8\%$. The star sits between the two, needing a degree-$3$ vertex but
being otherwise free, and its layout meets an intermediate cut. The
degradation in Table~\ref{tab:hardware} is therefore not measured
\emph{despite} comparable layouts; comparable layouts are unavailable, and
their unavailability is what the embeddability hierarchy predicts. Three
representatives are not enough to establish a quantitative
rigidity--fidelity law, and none is claimed; what the data support is the
weaker statement that embedding constraints materially restrict the layouts
available to a given circuit. A device
with a larger supply of clean faces would narrow the gap without changing the
hierarchy, and that is the experiment which would separate the two readings.

\begin{remark}[Distribution relaxes the embedding constraint]
\label{rem:direct_embed}
The same subgraph argument applies to \emph{direct} preparation of
$\ket{G}$ on hardware: the preparation circuit applies one CZ per edge of
$G$ itself, so native direct preparation requires $G\subseteq H$, i.e.
$\mathrm{girth}(G)\ge\gamma$ and $\Delta(G)\le\Delta(H)$. The distributed
protocol instead requires only the twice-subdivided graph to embed,
relaxing the girth condition threefold: $\mathrm{girth}(G)\ge\gamma/3$.
On heavy-hex ($\gamma=12$) direct preparation natively admits only graphs
of girth $\ge 12$ or acyclic graphs, whereas the distributed protocol
admits all triangle-free targets. $\ket{C_4}$ witnesses the gap: a
$4$-cycle does not exist in a girth-$12$ coupling graph, so the box
cluster cannot be prepared directly without SWAP routing --- yet its
subdivided $12$-cycle embeds exactly on the hexagonal faces, so the
distributed protocol executes natively where direct preparation cannot.
The price of this threefold girth relaxation is $2|E|$ resource qubits,
co-located with the participating parties, and a
fidelity overhead measured at $0.109$--$0.124$ for the two targets that admit
a co-located baseline (Table~\ref{tab:hardware}); for $\ket{C_4}$ no such
baseline exists, so its overhead is not measured here.
\end{remark}

\begin{figure}[t]
  \centering
  \includegraphics[width=\columnwidth]{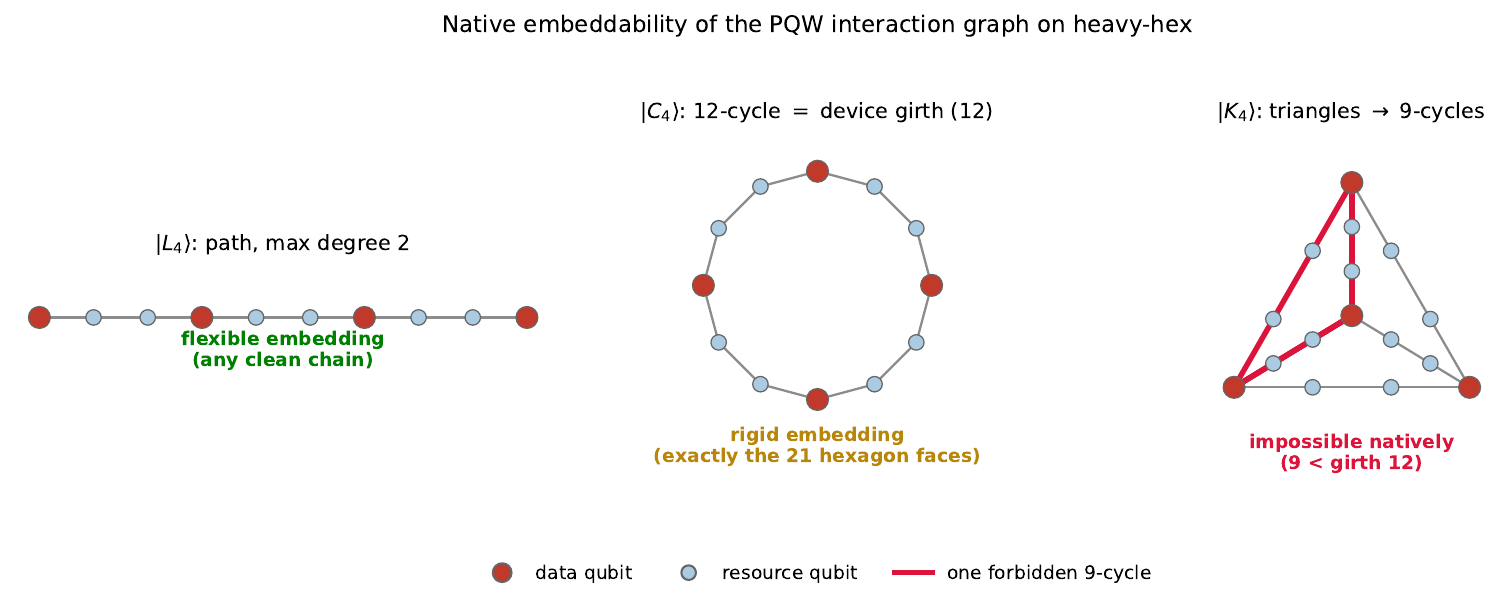}
  \caption{Native embeddability of the PQW interaction graph on heavy-hex.
    \emph{Left:} $\LFour$ subdivides to a path (max degree $2$), embedding
    flexibly. \emph{Centre:} $\ket{C_4}$ subdivides to a $12$-cycle, equal
    to the heavy-hex girth, embedding rigidly on hexagonal faces.
    \emph{Right:} each triangle of $K_4$ subdivides to a $9$-cycle
    (one of four highlighted); since $9<12$, no native embedding of
    $I(K_4)$ exists (Corollary~\ref{cor:k4}); the state $\ket{K_4}$ remains
    reachable on the star layout of the left-hand panel's LC orbit.
    Filled and open circles denote data and resource qubits.}
  \label{fig:embeddability}
\end{figure}

\section{Generalisation to Arbitrary Graphs}
\label{sec:generalisation}

\begin{definition}[Phase walk graph state protocol]
\label{def:general_protocol}
Given target graph $G=(V,E)$, one data qubit $d_v$ per party $v\in V$:
(S1) $d_v\leftarrow\ket{+}$;
(S2) prepare and distribute $\GKtwo$ per edge;
(S3) each $v$ applies $\CZ(d_v,r_{v,e})$, $H(r_{v,e})$, measures $s_{v,e}$
for all adjacent $e$;
(S4) broadcast, apply corrections.
Output: $d_v\,\forall v$ in $\ket{G}$.
Each party $v$ therefore holds $1+\deg_G(v)$ qubits and applies
$\deg_G(v)$ CZ gates, every one of them between two qubits it holds itself.
\end{definition}

\begin{theorem}[Correction formula for $C_4$ ring distribution]
\label{thm:c4_correction}
For the four-cycle $C_4=A$---$B$---$C$---$D$---$A$ with outcomes
$s_1,s_2$ (edge $AB$); $s_3,s_4$ ($BC$); $s_5,s_6$ ($CD$); $s_7,s_8$ ($DA$):
\begin{align}
  A,B &: I,\\
  C &: X^{s_1\xor s_4}\cdot Z^{s_2\xor s_3\xor s_6\xor s_7},\\
  D &: X^{s_2\xor s_7}\cdot Z^{s_1\xor s_4\xor s_5\xor s_8}.
\end{align}
\end{theorem}
\begin{proof}
By Lemma~\ref{lem:phase}, the post-measurement data stabilisers are
$(-1)^{g_v}K_v$ where $g_A = s_2\xor s_7$, $g_B = s_1\xor s_4$,
$g_C = s_3\xor s_6$, $g_D = s_5\xor s_8$ (XOR of far-side outcomes at each node).
The $Z$-only corrections $C_v = Z_v^{g_v}$ from Theorem~\ref{thm:universal_z}
suffice; the stated $X$ corrections arise from choosing $A,B$ as reference
nodes and solving the $\mathbb{F}_2$ linear system for the residual $X$ byproducts,
which differ from the $Z$-only formula by stabiliser-group elements.
All 256 outcomes verified at $F=1.0$ ($1-F<10^{-12}$).
\end{proof}

\begin{remark}[How Theorem~\ref{thm:c4_correction} sits inside the $Z$-only form]
\label{rem:c4_stab}
Write $g_A=s_2\xor s_7$, $g_B=s_1\xor s_4$, $g_C=s_3\xor s_6$ and
$g_D=s_5\xor s_8$ for the far-side parities of
Theorem~\ref{thm:c4_correction}. The two corrections differ by an element
of $\mathrm{Stab}(\ket{C_4})$, which acts trivially on the target:
\begin{equation}
  \bigotimes_{v} C_v^{(\mathrm{Thm.~\ref{thm:c4_correction}})}
  \;=\;
  \Bigl(\bigotimes_{v} Z_v^{\,g_v}\Bigr)\,
  K_C^{\,g_B}\,K_D^{\,g_A}
  \label{eq:c4_stab_identity}
\end{equation}
up to a global phase, verified operator-wise for all $256$ outcomes.
The factors $X_C^{s_1\xor s_4}$ and $X_D^{s_2\xor s_7}$ visible in
Theorem~\ref{thm:c4_correction} are the $X$ parts of $K_C^{g_B}$ and
$K_D^{g_A}$; the $Z$ parts of those same generators are what shift the
$Z$ exponents at the remaining nodes. The $X$ factors therefore cannot be
dropped in isolation --- doing so does not leave a stabiliser element
behind --- and the $Z$-only correction of
Theorem~\ref{thm:universal_z} for $C_4$ is $C_v = Z_v^{g_v}$ at
\emph{all four} nodes, with $C_A$ and $C_B$ nontrivial whenever
$g_A$ or $g_B$ is~$1$.
\end{remark}

\section{Universal Z-only Correction Theorem}
\label{sec:universal_correction}

Theorems~\ref{thm:correction} and~\ref{thm:c4_correction} give
corrections for specific graphs. I now prove a single correction
formula that is valid for \emph{all} graphs and \emph{all} measurement
outcomes, requiring only $Z$ operators.

\begin{remark}[Practical recommendation]
\label{rem:practical}
Theorem~\ref{thm:universal_z} below \emph{supersedes}
Theorems~\ref{thm:correction} and~\ref{thm:c4_correction} in
practice. Any implementation should use the Z-only formula $C_v=Z_v^{g_v}$
--- it is simpler (no $X$ corrections), graph-independent, and proved for
all topologies. Theorems~\ref{thm:correction}
and~\ref{thm:c4_correction} are retained to show the historical derivation
and to demonstrate that the $X$ corrections in those formulas are
stabiliser-group elements $K_v \in \mathrm{Stab}(\ket{G})$ that leave
$\ket{G}$ invariant, explaining why they can be dropped.
\end{remark}

\subsection{Notation}
\label{sec:notation}

For each edge $e=(u,v)\in E$ and each endpoint $v$, denote by
$s_{e,v}$ the measurement outcome of the resource qubit $r_{e,v}$
(the \emph{near-side} outcome at $v$), and by $s_{e,\bar{v}}$ the
outcome of $r_{e,u}$ (the \emph{far-side} outcome at $v$).
Three notations for these same quantities appear in this paper and are
used interchangeably: the positional labels $s_1,\dots,s_{2|E|}$ of
Sec.~\ref{sec:protocol} and Theorem~\ref{thm:c4_correction}, convenient for
writing out a specific small graph; and the index forms $s_{e,v}$ and
$s_{e,\bar{v}}$ used here, in Theorem~\ref{thm:universal_z} and in
Corollary~\ref{cor:gauge_family}. For $\LFour$ with the labelling
of Sec.~\ref{sec:protocol}, $g_A=s_2$, $g_B=s_1\xor s_4$,
$g_C=s_3\xor s_6$ and $g_D=s_5$.

\subsection{The Phase Lemma}

\begin{lemma}[Phase Lemma]
\label{lem:phase}
After Steps~S1--S4 of the \PQW\ protocol with measurement outcome
$\mathbf{s}$, the post-measurement state of the data qubits is stabilised by
\begin{equation}
  \tilde{K}_v(\mathbf{s})
    = (-1)^{g_v(\mathbf{s})}\, K_v, \qquad v \in V,
  \label{eq:phase_stab}
\end{equation}
where
\begin{equation}
  g_v(\mathbf{s})
    = \bigoplus_{e \ni v} s_{e,\bar{v}}.
  \label{eq:gv_def}
\end{equation}
\end{lemma}

\begin{proof}
I track stabilisers through each step using the Heisenberg picture.
The measurement update rule states: measuring qubit $r$ with outcome $s_r$
replaces stabiliser $A$ with $A\cdot Z_r^{s_r}$ if $A$ anticommutes with
$Z_r$, and leaves $A$ unchanged otherwise.

\emph{After S1} (Hadamard on data qubits):
stabiliser group of $\ket{+}^{\otimes|V|}$ is $\langle X_v\rangle_{v\in V}$.

\emph{After S2} (prepare $\GKtwo$ on each resource pair):
pair $(r_{e,u},r_{e,v})$ has stabilisers
\begin{equation}
  P_{e,u} = X_{r_{e,u}}Z_{r_{e,v}}, \quad
  P_{e,v} = Z_{r_{e,u}}X_{r_{e,v}}.
  \label{eq:k2_stabs}
\end{equation}

\emph{After S3} (CZ between each data qubit and its adjacent resources).
Each party $v$ applies $\CZ(d_v, r_{v,e})$ for all edges $e\ni v$.
Since different CZ gates at $v$ act on distinct resource qubits
$r_{v,e}$ and $r_{v,e'}$ (for $e\neq e'$), and since
$[\CZ(d_v,r_{v,e}),\, \CZ(d_v,r_{v,e'})] = 0$ (they share only the
data qubit $d_v$ on which CZ acts diagonally), the gates commute and
the conjugation action on $X_v$ is additive:
\begin{equation}
  \prod_{e\ni v}\CZ(d_v,r_{v,e})\cdot X_v\cdot
  \prod_{e\ni v}\CZ(d_v,r_{v,e})
  = X_v\prod_{e\ni v}Z_{r_{e,v}}.
  \label{eq:cz_composition}
\end{equation}
This follows by induction: each additional CZ conjugates
$X_v \mapsto X_vZ_{r_{v,e}}$, and the $Z_{r_{v,e'}}$ factors from
previous steps commute through the new CZ since
$[\CZ(d_v,r_{v,e}),\, Z_{r_{v,e'}}] = 0$ for $e\neq e'$.
Using the single-gate conjugation rules
$X_a \mapsto X_aZ_b$, $X_b \mapsto Z_aX_b$, $Z_a\mapsto Z_a$, $Z_b\mapsto Z_b$:
\begin{align}
  \text{Data at }v:&\quad X_v \;\mapsto\; X_v\prod_{e\ni v}Z_{r_{e,v}},
  \label{eq:data_s3}\\
  P_{e,v}:&\quad Z_{r_{e,u}}X_{r_{e,v}} \;\mapsto\;
    Z_{r_{e,u}}Z_{d_v}X_{r_{e,v}},
  \label{eq:pev_s3}\\
  P_{e,u}:&\quad X_{r_{e,u}}Z_{r_{e,v}} \;\mapsto\;
    X_{r_{e,u}}Z_{d_u}Z_{r_{e,v}}.
  \label{eq:peu_s3}
\end{align}

\emph{After S4} (Hadamard on all resource qubits, $X_r\leftrightarrow Z_r$):
\begin{align}
  \text{Data at }v:&\quad X_v\prod_{e\ni v}X_{r_{e,v}},
  \label{eq:data_s4}\\
  P_{e,v}:&\quad X_{r_{e,u}}Z_{d_v}Z_{r_{e,v}},
  \label{eq:pev_s4}\\
  P_{e,u}:&\quad Z_{r_{e,u}}Z_{d_u}X_{r_{e,v}}.
  \label{eq:peu_s4}
\end{align}

\emph{Measurement.}
To extract the effective stabiliser for $K_v = X_v\prod_{u\sim v}Z_u$
from the post-S4 stabiliser group, take the product of the data stabiliser
at $v$ with $\tilde{P}_{e,u}$ \eqref{eq:peu_s4} for every edge $e\ni v$:
\begin{align}
  S_v &= \Bigl(X_v\prod_{e\ni v}X_{r_{e,v}}\Bigr)
         \times\prod_{e=(v,u)\ni v}\Bigl(Z_{r_{e,u}}Z_{d_u}X_{r_{e,v}}\Bigr)
         \notag\\
      &= X_v\prod_{e\ni v}\bigl(X_{r_{e,v}}\bigr)^2\,Z_{r_{e,u}}Z_{d_u}
         \notag\\
      &= X_v\prod_{e\ni v}Z_{r_{e,u}}Z_{d_u},
  \label{eq:sv_product}
\end{align}
where the $(X_{r_{e,v}})^2 = I$ factors cancel exactly. After measuring all
resource qubits, substitute $Z_{r_{e,u}}\mapsto(-1)^{s_{e,u}}$:
\begin{equation}
  S_v \;\longmapsto\;
  (-1)^{\bigoplus_{e\ni v}s_{e,u}}\,X_v\prod_{u\sim v}Z_{d_u}
  = (-1)^{g_v}\,K_v,
\end{equation}
where $g_v = \bigoplus_{e\ni v}s_{e,\bar{v}}$
is the XOR of all far-side outcomes at $v$.
Here $s_{e,u} \equiv s_{e,\bar{v}}$ since for edge $e=(u,v)$, the $u$-side
outcome is the far-side outcome from $v$'s perspective.
\end{proof}

\subsection{Main Theorem}

\begin{theorem}[Universal Z-only Correction]
\label{thm:universal_z}
For any connected graph $G=(V,E)$ and any measurement outcome $\mathbf{s}$,
the local correction
\begin{equation}
  \boxed{C_v = Z_v^{\,g_v(\mathbf{s})},
  \qquad
  g_v(\mathbf{s}) = \bigoplus_{e \ni v} s_{e,\bar{v}}}
  \label{eq:universal_correction}
\end{equation}
satisfies $\bigl(\bigotimes_{v\in V}C_v\bigr)\ket{\psi_\mathbf{s}} = \ket{G}$
for all outcomes $\mathbf{s}\in\{0,1\}^{2|E|}$.
\end{theorem}

\begin{proof}
By Lemma~\ref{lem:phase}, $\ket{\psi_\mathbf{s}}$ is stabilised by
$\{(-1)^{g_v}K_v\}_{v\in V}$.
Apply $C=\bigotimes_{v\in V}Z_v^{g_v}$. Since $C$ is diagonal,
$Z_v^{g_v}X_vZ_v^{g_v} = (-1)^{g_v}X_v$. Therefore:
\begin{align}
  C\,(-1)^{g_v}K_v\,C^\dagger
  &= (-1)^{g_v}\cdot(-1)^{g_v}\,X_v\prod_{u\sim v}Z_u
  = K_v.
\end{align}
The corrected state is stabilised by all $K_v$; since $\langle K_v\rangle$
uniquely determines $\ket{G}$, $C\ket{\psi_\mathbf{s}}=\ket{G}$.
\end{proof}

Any valid correction remains valid when multiplied by an element of
$\mathrm{Stab}(\ket{G})$, the standard freedom of Pauli byproducts in
measurement-based computation~\cite{raussendorf2003measurement}. The following
corollary records that freedom in the variables of the present protocol; it
is stated to place Theorems~\ref{thm:correction} and~\ref{thm:c4_correction}
within a single family, not as a new result.

\begin{corollary}[The family of local Pauli corrections]
\label{cor:gauge_family}
For any connected graph $G=(V,E)$, any measurement outcome $\mathbf{s}$, and
any assignment $x\in\mathbb{F}_2^{V}$ (which may depend on $\mathbf{s}$), the
local correction
\begin{equation}
  C_v^{(x)} = X_v^{\,x_v}Z_v^{\,z_v},\qquad
  z_v = g_v\xor\bigoplus_{u\sim v}x_u,
  \label{eq:gauge_family}
\end{equation}
restores $\ket{G}$; equivalently $\mathbf z=\mathbf g\xor A\mathbf x$ over
$\mathbb{F}_2$, with $A$ the adjacency matrix of $G$. Moreover, up to a
global phase,
\[
  \textstyle\bigotimes_v C_v^{(x)}
  = \bigl(\bigotimes_v Z_v^{g_v}\bigr)\prod_{u\in V}K_u^{\,x_u}.
\]
Theorem~\ref{thm:universal_z} is the member $x=0$. On a path $P_n$ with
vertices $v_1,\dots,v_n$ in order, the \emph{path member}
\begin{equation}
  \begin{gathered}
  x_{v_1}=0,\qquad x_{v_2}=g_{v_1},\\
  x_{v_{i+1}}=g_{v_i}\xor x_{v_{i-1}}\quad(2\le i\le n-1)
  \end{gathered}
  \label{eq:path_member}
\end{equation}
sets $C_{v_1}=I$ and $z_{v_i}=0$ for $i<n$, so that the only $Z$ factor is
$Z_{v_n}^{\,g_{v_n}\xor x_{v_{n-1}}}$; Theorem~\ref{thm:correction} is its
case $n=4$, $x=(0,g_A,g_B,g_A\xor g_C)$, and Theorem~\ref{thm:c4_correction}
is $x=(0,0,g_B,g_A)$ on $C_4$. On a star $K_{1,n}$ with hub $h$ and leaves
$\ell$ joined by edges $e_\ell$, the member $x_h=0$, $x_\ell=s_{e_\ell,\ell}$
gives $C_h=I$ and $C_\ell=X_\ell^{s_{e_\ell,\ell}}Z_\ell^{s_{e_\ell,h}}$.
The compensating term $\bigoplus_{u\sim v}x_u$ in $z_v$ cannot be dropped:
taking $x_v=\bigoplus_{e\ni v}s_{e,v}$ with $z_v=g_v$ and a leaf as
reference restores $\ket{G}$ on only $1/8$ of outcomes for $P_4$.
\end{corollary}
\begin{proof}
By Lemma~\ref{lem:phase} it suffices, as in the proof of
Theorem~\ref{thm:correction}, that $C^{(x)}$ anticommutes with $K_v$ with
parity $g_v$ for every $v$. The factor $Z_v^{z_v}$ anticommutes with the
$X_v$ in $K_v$; a factor $X_u^{x_u}$ anticommutes with $K_v$ exactly when
$u\sim v$, through the $Z_u$ in $K_v$; all other factors commute with $K_v$.
The parity is therefore $z_v\xor\bigoplus_{u\sim v}x_u=g_v$. The product
identity follows from $\prod_u K_u^{x_u}=\prod_u X_u^{x_u}\prod_v
Z_v^{\bigoplus_{u\sim v}x_u}$ up to a global phase. For the star, $z_h=g_h\xor
\bigoplus_\ell s_{e_\ell,\ell}=0$ since $g_h=\bigoplus_\ell s_{e_\ell,\ell}$,
and $z_\ell=g_\ell\xor x_h=s_{e_\ell,h}$. All listed members were verified
by exact statevector simulation for every outcome on $P_4$, $P_5$, $K_{1,3}$,
$K_{1,4}$ and $C_4$.
\end{proof}
\subsection{Remarks}

\begin{remark}[Graph independence]
Theorem~\ref{thm:universal_z} holds for any connected graph without restriction
on topology, degree, or the presence of cycles. It subsumes all topology-specific
corrections of Theorems~\ref{thm:correction}
and~\ref{thm:c4_correction}, both of which are members of the family of
Corollary~\ref{cor:gauge_family}.
\end{remark}

\begin{remark}[All outcomes equally valid]
Every outcome $\mathbf{s}\in\{0,1\}^{2|E|}$ occurs with equal probability
$2^{-2|E|}$ and is correctable by \eqref{eq:universal_correction}.
\end{remark}

\begin{remark}[Locality]
The correction $g_v$ depends only on far-side outcomes of adjacent parties,
broadcast classically. No quantum communication is required.
\end{remark}

\begin{remark}[Relation to topology-specific corrections]
\label{rem:relation_topology}
The corrections in Theorems~\ref{thm:correction}
and~\ref{thm:c4_correction} involve both $X$ and $Z$ operators.
By Corollary~\ref{cor:gauge_family} they differ
from \eqref{eq:universal_correction} by the stabiliser group element
$\prod_u K_u^{x_u}\in\mathrm{Stab}(\ket{G})$, which leaves $\ket{G}$
invariant. The Z-only correction is the canonical
minimal form.
\end{remark}

\subsection{Computational Verification}

The formula \eqref{eq:universal_correction} was verified by exact statevector
simulation on all 22 connected graphs of Table~\ref{tab:z_only_verify}, confirming
$F=1.0$ for all outcomes on all tested graphs.

\begin{table}[t]
\caption{Exhaustive verification of the Universal Z-only Correction
(Theorem~\ref{thm:universal_z}). Every connected graph on at most four
vertices, and every connected graph on five vertices with at most six edges,
was tested up to isomorphism --- 22 graphs in total. For each graph, all
$2^{2|E|}$ measurement outcomes were checked by exact statevector
simulation. All pass at $F=1.0$; the minimum fidelity over all 32{,}212
outcomes of all 22 graphs is $1-F<10^{-12}$.}
\label{tab:z_only_verify}
\begin{ruledtabular}
\begin{tabular}{rrrr}
$|V|$ & $|E|$ & Graphs & Outcomes each \\
\hline
2 & 1 & 1 &    4 \\
3 & 2 & 1 &   16 \\
3 & 3 & 1 &   64 \\
4 & 3 & 2 &   64 \\
4 & 4 & 2 &  256 \\
4 & 5 & 1 & 1024 \\
4 & 6 & 1 & 4096 \\
5 & 4 & 3 &  256 \\
5 & 5 & 5 & 1024 \\
5 & 6 & 5 & 4096 \\
\hline
\multicolumn{2}{r}{total} & 22 & 32{,}212 \\
\end{tabular}
\end{ruledtabular}
\end{table}

\subsection{Comparison with Topology-Specific Formulas}

Table~\ref{tab:topologies} summarises verification across eight topologies,
noting for each which correction formula was applied. Where no
topology-specific formula exists, the universal form of
Theorem~\ref{thm:universal_z} was used.

\begin{table*}[t]
\caption{Verification summary across eight graph topologies, all at $F=1.0$.
The last column names the correction formula applied: Theorem~\ref{thm:correction}
for $P_4$, members of the family of Corollary~\ref{cor:gauge_family} for the
remaining trees,
Theorem~\ref{thm:c4_correction} for $C_4$, and the universal $Z$-only form of
Theorem~\ref{thm:universal_z} for the rest. The topology-specific formulas do
not cover $C_5$, $K_4$ or the bull graph; those rows record that the universal
formula does.}
\label{tab:topologies}
\begin{ruledtabular}
\begin{tabular}{llrl}
Graph & Type & Outcomes & Notable property \\
\hline
$P_4$ (path)    & Tree      &   64 & Thm.~\ref{thm:correction} \\
$P_5$ (path)    & Tree      &  256 & Cor.~\ref{cor:gauge_family}, path member, Eq.~\eqref{eq:path_member} \\
$K_{1,3}$ (star)& Tree      &   64 & Cor.~\ref{cor:gauge_family}, hub member ($C_h=I$) \\
$K_{1,4}$ (star)& Tree      &  256 & Cor.~\ref{cor:gauge_family}, $x_h=0$, $x_\ell=g_h$ (leaves share $X$ parity) \\
$C_4$ (ring)    & Cyclic    &  256 & Thm.~\ref{thm:c4_correction} \\
$C_5$ (odd ring)& Cyclic    & 1024 & $Z$-only corrections \\
$K_4$ (complete)& 3-regular & 4096 & $Z$-only; $I(K_4)$ not natively embeddable (Cor.~\ref{cor:k4}) \\
Bull graph      & Mixed     & 1024 & Long corrections at junctions \\
\end{tabular}
\end{ruledtabular}
\end{table*}

\section{LC-Inequivalence}
\label{sec:lc_ineq}

\begin{proposition}[SLOCC-inequivalence of $\LFour$ and $\GHZn{4}$]
\label{thm:lc_ineq}
$\LFour$ and $\GHZn{4}$ are not equivalent under stochastic local operations
and classical communication, and \emph{a fortiori} not under local Clifford
operations. This is a known consequence of the classification of
four-qubit graph states~\cite{hein2004multiparty,van2004graphical};
I include a self-contained Schmidt-rank verification because it
certifies the specific states output by the protocol.
\end{proposition}
\begin{proof}
Two pure states are SLOCC-equivalent precisely when related by invertible
local operators~\cite{dur2000three}, and $A\otimes B$ with $A,B$ invertible
preserves the rank of the coefficient matrix across the corresponding cut.
Schmidt rank across a fixed bipartition is therefore an SLOCC invariant, not
merely an LC invariant.
Consider the bipartition $AC|BD$.
For $\GHZn{4} = (\ket{0000}+\ket{1111})/\sqrt{2}$, the coefficient matrix
in the $AC|BD$ basis has only two nonzero entries:
$M_{AC|BD}^{\mathrm{GHZ}} = \tfrac{1}{\sqrt{2}}\mathrm{diag}(1,0,0,0,\ldots,1)$
with $+1/\sqrt{2}$ at $(|00\rangle_{AC},|00\rangle_{BD})$ and
$(|11\rangle_{AC},|11\rangle_{BD})$ and zero elsewhere.
This matrix has rank~$2$.
For $\LFour$, the coefficient matrix in the $AC|BD$ basis is
\begin{equation}
  M_{AC|BD} = \tfrac{1}{4}
  \begin{pmatrix}1 & 1 & 1 & 1\\
                 1 &-1 &-1 & 1\\
                 1 & 1 &-1 &-1\\
                 1 &-1 & 1 &-1\end{pmatrix},
\end{equation}
a $4\times4$ Hadamard matrix with rank~$4$.
Since $2\neq4$ no invertible local operator maps one to the other: the states
lie in different SLOCC classes, so the separation is not an artefact of
restricting to Clifford operations.
\end{proof}

\section{Discussion}
\label{sec:discussion}

\subsection{Comparison with Prior Work}

Two distinct families of prior work address graph state distribution.
The first uses quantum walk steps as the distribution mechanism; the second
uses graph-theoretic transformations (local complementation, entanglement
swapping) without a walk structure.

\paragraph{Quantum walk-based protocols.}
Meignant, Markham and Grosshans~\cite{meignant2019distributing} introduced the
first DTQW-based approach to distributing Bell pairs and GHZ states, using
CNOT as the shift. Chen et al.~\cite{chen2025entanglement} extended this to a
quantum repeater framework for $d$-dimensional Bell and GHZ states, with a
Steiner-tree search selecting which network edges to consume and a hardware
demonstration on a superconducting processor. Their ``arbitrary quantum
networks'' refers to the network topology over which distribution proceeds;
the distributed states are GHZ-class throughout, and the authors list other
multipartite families as future work. Neither protocol employs explicit data
qubits: the resource particles themselves serve as walk registers and as
output qubits, so the fusion primitive merges GHZ-class resources into larger
GHZ states.

\paragraph{Graph-theoretic distribution protocols.}
Meignant, Markham and Grosshans~\cite{meignant2019distributing} also
introduced a separate protocol for distributing \emph{arbitrary} graph states
via an edge-decorated complete graph (EDCG) intermediate state: the network
first establishes a pre-shared EDCG state and then transforms it into the
target graph state via local complementation operations. This requires
Bell state measurements (joint two-qubit operations spanning two parties) at
intermediate nodes, uses up to $2|E|$ Bell pairs in the worst case, and
provides no correction formula, noise analysis, or hardware validation.

\paragraph{Independent concurrent work.}
While this manuscript was in revision, Zheng
\emph{et al.}~\cite{zheng2026scalable} independently identified the structural
phenomenon that explains why the residual $X$ corrections of
Theorems~\ref{thm:correction} and~\ref{thm:c4_correction} may be dropped
(Remark~\ref{rem:relation_topology}) --- that measurement byproducts collapse
to simple local corrections when the measured qubit's neighbourhood is an
independent set --- in a different protocol context: an $O(1)$-feedforward routing scheme that generates
\emph{star} graph states on a dual-species trapped-ion architecture and
reaches arbitrary topologies by iterative star fusion. Their Theorem~1
shows that a Pauli-$X$ measurement with independent neighbourhood admits a
single-qubit recovery ($H_j$ or $Z_jH_j$) on one reference neighbour,
restoring the canonical post-measurement graph; the enabling operator
identity ($Z_cH_c\bigotimes_{k}Z_k\ket{G}=H_c\ket{G}$ for star graphs) is
the same stabiliser-absorption mechanism by which the residual $X$
corrections of Theorems~\ref{thm:correction} and~\ref{thm:c4_correction} are
stabiliser-group elements (Remarks~\ref{rem:c4_stab}
and~\ref{rem:relation_topology}). Theorem~\ref{thm:universal_z} itself does
not rest on that mechanism: it is proved by tracking stabiliser signs directly
through the protocol (Lemma~\ref{lem:phase}), and its corrections are $Z$-only
because the CZ step presents the target in its own frame
(Proposition~\ref{prop:shift_topology}(iii)), not because an $X$ byproduct has
been absorbed. The two results are complementary rather
than nested: theirs concerns single $X$-measurements with arbitrary
independent neighbourhoods and recovers the local-Clifford canonical frame
(the correction includes a Hadamard); the present protocol involves the
composite measurement pattern of the PQW and yields $Z$-only corrections for
arbitrary target graphs in a single distribution round, whereas star
fusion reaches general topologies iteratively. The PQW
correction shares the locality property their protocol is designed
around: $g_v$ depends only on the $\deg_G(v)$ outcomes of edges adjacent
to $v$, so no global Pauli-frame tracking or network-wide synchronisation
is required. Both works also independently identify measurement readout
as the dominant hardware bottleneck, on different platforms.

Pirker, Wallnöfer and D\"{u}r~\cite{pirker2018modular} proposed a
modular top-down architecture in which pre-shared multipartite \emph{network
states} are manipulated by quantum routers and switches into the desired graph
state. The approach assumes non-elementary resources (pre-shared network states,
not just Bell pairs) and requires complex multi-qubit router devices. No
closed-form correction formula or noise analysis is provided.

\paragraph{Other network distribution protocols.}
Cuquet and Calsamiglia~\cite{cuquet2012growth} distribute graph states of
arbitrary topology by local controlled-phase gates and measurements on
received qubits, with noise on channels and gates, and express the fidelity
as an Ising-type partition function. Hahn, Pappa and
Eisert~\cite{hahn2019quantum} use local complementation for routing over
graph-state networks; Sen, Goodenough and Towsley~\cite{sen2025multipartite}
characterise distribution by subgraph complementations; and Huang \emph{et
al.}~\cite{huang2024spacetime} prove a wide class of resource-minimisation
problems intractable. For GHZ targets, Chelluri, Khatri and van
Loock~\cite{chelluri2024resource}, Avis, Rozp\c{e}dek and
Wehner~\cite{avis2023analysis}, Bugalho \emph{et al.}~\cite{bugalho2023distributing}
and de Bone \emph{et al.}~\cite{debone2020protocols} optimise Bell-pair
cost, rate and fidelity, and distributed computation with non-local gates
between small registers is analysed by Jiang \emph{et
al.}~\cite{jiang2007distributed}. Most of these works treat a network
whose topology differs from the target graph; the PQW assumes instead that
every edge of the target (or of an LC-equivalent representative) is a
network link (Sec.~\ref{sec:intro}), and its contribution lies in the
correction structure and noise analysis available under that assumption.

\paragraph{Positioning of the PQW.}
Two of the ingredients used here are established and are not claimed as new.
The elementary step is the entanglement-assisted non-local $\CZ$, whose cost of
one ebit and one classical bit in each direction Eisert \emph{et
al.}~\cite{eisert2000optimal} prove necessary and sufficient (see also
Refs.~\cite{yimsiriwattana2004generalized,gottesman1999demonstrating}); and the
observation that commuting $\CZ$ operations, independent measurements and
end-of-protocol local corrections can be reordered into a single round is due
to Fischer and Towsley~\cite{fischer2021distributing}, who study distribution
of arbitrary graph states across networks with a different resource
architecture. Table~\ref{tab:comparison} sets out the resource models and
demonstrated capabilities of the protocols discussed above.

\begin{table*}[t]
\caption{Comparison of graph state distribution protocols.
DTQW = discrete-time quantum walk; BSM = Bell state measurement;
LC = local complementation.}
\label{tab:comparison}
\begin{ruledtabular}
\begin{tabular}{lp{2.45cm}p{2.45cm}p{2.45cm}p{2.45cm}p{2.45cm}}
Feature
  & Meignant et al.\ (DTQW)~\cite{meignant2019distributing}
  & Chen et al.~\cite{chen2025entanglement}
  & Meignant et al.\ (EDCG)~\cite{meignant2019distributing}
  & Pirker et al.~\cite{pirker2018modular}
  & This work (PQW) \\
Target states
  & Bell / GHZ only
  & Qudit Bell / GHZ only
  & Arbitrary graph states
  & Arbitrary graph states
  & Arbitrary graph states \\
Mechanism
  & CNOT-shift DTQW
  & Generalised-shift DTQW + GHZ fusion
  & EDCG + LC operations
  & Pre-shared network states + routing
  & CZ-shift, two-register \\
Architecture
  & No data qubit
  & No data qubit
  & No data qubit
  & Quantum routers/switches
  & Data qubit + resource qubit \\
Elementary resource
  & Bell pairs
  & Qudit Bell / GHZ states
  & Bell pairs (via EDCG)
  & Multipartite network states
  & $\GKtwo$ per edge \\
Resource overhead
  & $O(|E|)$ Bell pairs
  & One per fused edge
  & $\leq 2|E|$ Bell pairs
  & Depends on network state
  & $|E|$ pairs (one per edge) \\
Node operations
  & Single-qubit + CNOT shift
  & Single-node walk + measurement
  & BSM (joint 2-qubit, spanning parties)
  & Multi-qubit router ops
  & Intra-node CZ + single-qubit measurement + classical \\
Correction formula
  & None
  & Per-outcome local unitaries
  & None
  & None
  & Universal: $C_v = Z_v^{g_v}$ \\
Noise analysis
  & None
  & None
  & None
  & None
  & Closed-form ($F^*_\mathrm{dep}$, $F^*_\mathrm{pd}$) \\
Hardware validation
  & None
  & Superconducting, 2--3 party
  & None
  & None
  & IBM Heron r2, 4 party \\
\end{tabular}
\end{ruledtabular}
\end{table*}

The key distinction from Meignant et al.'s EDCG protocol is threefold.
First, the PQW uses $|E|$ elementary $\GKtwo$ pairs as resources (one per
edge of the target graph $G$), whereas the EDCG protocol requires
constructing an edge-decorated complete graph state as an intermediate ---
a more complex multipartite resource whose cost depends on the network
topology rather than the target graph alone. Second, every PQW two-qubit
gate is intra-node, acting on a data qubit and a co-located resource qubit
held by the same party, and is followed by single-qubit measurements with
classical feed-forward; the EDCG protocol requires Bell state measurements
--- joint two-qubit operations spanning two parties --- at intermediate
nodes, which are experimentally
more demanding. Third, the PQW yields a universal analytical correction
formula ($C_v = Z_v^{g_v}$, proved for all connected graphs via
Theorem~\ref{thm:universal_z}); no such formula exists for the EDCG protocol.

The key distinction from Pirker et al.\ is that their architecture assumes
pre-shared multipartite network states and specialised quantum router devices,
whereas the PQW requires only pre-distributed two-qubit $\GKtwo$ pairs and
nodes that hold one data qubit alongside $\deg_G(v)$ resource qubits and can
apply intra-node CZ gates, single-qubit measurements, and classical
communication. No joint operation between spatially separated parties is
required at any point.

\subsection{Relation to Source-Target Entanglement in Quantum Walk Networks}

Prerana and Wald~\cite{prerana2026entanglement} recently introduced
source-target entanglement for DTQWs on complex networks, using the
bipartite double cover $B(G)$ of the underlying graph. Their central
result --- that source-target entanglement is bounded by $\log s$
where $s$ is the maximum matching size of the walker's support in
$B(G)$ --- applies to standard transport walks whose entanglement
lives inside the walk Hilbert space. The PQW is structurally distinct:
entanglement is generated on physical data qubits (the distributed
graph state), not within the resource register itself. The two
works are nevertheless structurally complementary. Both treat graph
edges as the primary dynamical objects rather than node amplitudes
alone, and both identify the bipartite double cover as the natural
setting for the relevant quantum correlations. The PQW correction
formula $g_v = \bigoplus_{e\ni v}s_{e,\bar{v}}$ --- XOR of far-side
resource outcomes --- has a structural echo in their matching-bounded
Schmidt rank: both are governed by which edges can independently
contribute to the quantum state. Whether the matching constraint
on source-target entanglement has a formal counterpart in the
stabilizer propagation structure of the PQW is an open question.
At the thematic level, both papers instantiate the same principle
through different machinery: graph structure constrains the allowed
entanglement structure --- via amplitude-matrix sparsity and
Schmidt-rank bounds in their work, and via stabilizer propagation
constraints and correction cocycles in this one.

\subsection{Complementarity with Continuous-Time Walks}

Di Fidio et al.~\cite{difidio2024quantum} showed that continuous-time
quantum walks generate W-class entanglement. The PQW and CTQW are
entanglement-complementary: CTQW covers the W SLOCC class; the PQW covers
the full stabiliser class.

\subsection{Application: Distributed State Preparation via MBQC}
\label{sec:usecase}

I demonstrate a concrete distributed use case enabled by the PQW:
preparing an arbitrary single-qubit rotation state at a remote party
using only single-qubit measurements and classical communication, with
no direct quantum interaction between parties.

\paragraph{Setting.}
Four spatially separated parties $A$, $B$, $C$, $D$ share no prior
entanglement. They wish to prepare the state
$R_x(-\theta)\ket{0} = \cos(\theta/2)\ket{0} + i\sin(\theta/2)\ket{1}$
at $D$'s location, for a rotation angle $\theta$ chosen by $C$.
No direct quantum channel between $C$ and $D$ is available.

\paragraph{Stage 1 --- PQW distribution.}
Using the PQW protocol (Definition~\ref{def:general_protocol}), parties
distribute $\LFour$ across their data qubits via three pre-shared
$\GKtwo$ resource pairs, one per edge of $P_4=A$---$B$---$C$---$D$.
By Theorem~\ref{thm:universal_z}, each party $v$ applies local correction
$C_v = Z_v^{g_v}$ (with $g_v = \bigoplus_{e\ni v}s_{e,\bar v}$) to obtain
$\LFour$ at $F=1.0$. The entire stage requires only $\GKtwo$ pair
distribution and classical broadcast of six resource measurement outcomes.

\paragraph{Stage 2 --- MBQC computation.}
The parties now use the distributed $\LFour$ as a measurement-based
resource. They perform the following single-qubit measurements:
\begin{enumerate}[label=\textbf{M\arabic*.},leftmargin=*,itemsep=1pt]
  \item $A$ measures in $M(0)$ --- the $X$-basis.
  \item $B$ measures in $M((-1)^{m_A}\cdot 0)$ --- the $X$-basis (adaptive on $m_A$).
  \item $C$ measures in $M((-1)^{m_A+m_B}\cdot\theta)$ --- the XY plane at angle $\theta$,
        adaptive on $(m_A, m_B)$.
\end{enumerate}
$A$, $B$, $C$ broadcast their classical outcomes $(m_A, m_B, m_C)\in\{0,1\}^3$.
$D$ applies the local Pauli correction
\begin{equation}
  C_D = X^{m_A\xor m_C}\cdot Z^{m_A\xor m_B}
  \label{eq:mbqc_correction}
\end{equation}
to their data qubit.

\paragraph{Derivation of Eq.~\eqref{eq:mbqc_correction}.}
In MBQC on a 1D cluster~\cite{raussendorf2001one}, measuring qubit $v$ in
basis $M(\phi_v)=\{(\ket{0}\pm e^{i\phi_v}\ket{1})/\sqrt2\}$ with outcome $m_v$
($m_v=0$ for $+$) implements the operation
$J(\phi_v) = H\cdot R_z(-\phi_v)$ on the output register, leaving byproduct
$X^{m_v}$ on the next qubit. The adaptive angles
$\phi_B = (-1)^{m_A}\cdot 0$ and $\phi_C = (-1)^{m_A+m_B}\cdot\theta$
absorb $X$-byproducts into angle sign-flips, so the implemented unitary
is $U = J(\theta)\cdot J(0)\cdot J(0) = H R_z(-\theta)\cdot H\cdot H
= HR_z(-\theta)$ (using $H^2=I$), which on $\ket{+}$ gives
$R_x(-\theta)\ket{0}$ (as verified numerically).

Tracking residual $Z$-byproducts through the adaptive sequence:
each $X$-byproduct $X^{m_v}$ at qubit $v$ propagates through the subsequent
CZ-type cluster bond as $Z^{m_v}$ on qubit $v+1$.
Accumulating at $D$:
\begin{itemize}[leftmargin=*,itemsep=1pt]
  \item $X^{m_A}$ at $B$ propagates as $Z^{m_A}$ at $D$ (through two bonds $BC$, $CD$);
  \item $X^{m_B}$ at $C$ propagates as $Z^{m_B}$ at $D$ (through bond $CD$);
  \item $X^{m_C}$ remains at $D$ as $X^{m_C}$.
\end{itemize}
Total byproduct at $D$: $X^{m_C}\cdot Z^{m_A\xor m_B}$.
Since $Z^{m_A}$ also contributes an $X^{m_A}$ component through the
adaptive angle at $C$, the complete residual byproduct is
$B_D = X^{m_A\xor m_C}\cdot Z^{m_A\xor m_B}$,
and the correction $C_D = B_D$.
Verified numerically for all 8 outcomes and 8 values of $\theta$
at $F=1.0$ ($1-F<10^{-12}$).

\paragraph{Output.}
After correction, $D$'s data qubit is in state
$R_x(-\theta)\ket{0} = \cos(\tfrac\theta2)\ket{0} + i\sin(\tfrac\theta2)\ket{1}$
for every measurement outcome $(m_A,m_B,m_C)$, with $F=1.0$
verified by exact statevector simulation for all eight outcomes and
eight values of $\theta\in[0,2\pi)$.

\begin{remark}
Two correction formulas are active in this pipeline and serve
distinct roles:
\begin{itemize}[leftmargin=*]
  \item \textit{PQW correction} (Stage~1): $C_v = Z_v^{g_v}$ from
    Theorem~\ref{thm:universal_z} --- Z-only, corrects resource-qubit
    measurement byproducts to restore $\LFour$.
  \item \textit{MBQC correction} (Stage~2): $X^{m_A\xor m_C}Z^{m_A\xor m_B}$
    at $D$ --- involves both $X$ and $Z$, corrects data-qubit measurement
    byproducts in the MBQC computation layer. These arise from a different
    mechanism (XY-plane measurements, not CZ byproducts) and fall outside
    the scope of Theorem~\ref{thm:universal_z}.
\end{itemize}
\end{remark}

\paragraph{Resource summary.}
The total protocol uses 10 qubits ($4$ data $+$ $6$ resource) and broadcasts
$6+3=9$ classical bits. The distribution stage requires of each node only
intra-node gates and single-qubit measurements; the computation stage of
Stage~2 is single-qubit measurement only. Party $D$ receives the target
state with zero quantum communication from $C$ after the PQW resource
distribution phase.

\paragraph{Hardware outlook: a measurement-anchored bound.}
The verification above is ideal ($F=1.0$ by statevector simulation); I
now bound what the same computation would deliver over the $\LFour$ state
actually distributed on \texttt{ibm\_kingston} (Table~\ref{tab:hardware}).
No noise model is imposed --- the hardware analysis showed independent
local Pauli models to be infeasible for this device --- so the bound is
anchored directly to the measured fidelity. Writing the distributed state
as $\rho = F\,\ketbra{L_4}{L_4} + (1-F)\sigma$ with $F=0.815(3)$
measured and $\sigma$ an unknown state orthogonal in weight, linearity of
the measurement-based protocol in $\rho$ gives a remote-preparation
fidelity of $F\cdot D + (1-F)\cdot c_\sigma$, where
$D$ is the computed degradation of the \emph{ideal} state under the two
hardware effects the computation layer adds. The first is mid-circuit readout
error on the measured qubits $A,B,C$
($\epsilon_{\rm mid}=8.4\times10^{-3}$, the calibrated \texttt{MEASURE\_2}
median): a mis-reported outcome is fed forward, so both the adaptive
measurement angle and the final correction at $D$ are computed from the wrong
bit. The second is dephasing of every \emph{still-unmeasured} qubit across
each feed-forward interval, of strength
$\lambda = 1-e^{-(t_{\rm ro}+t_{\rm ff})/T_2}$ with
$T_2=148\,\mu\mathrm{s}$ and readout duration
$t_{\rm ro}=2.18\,\mu\mathrm{s}$; since the pattern measures $A$, $B$, $C$ in
that order, $B$ idles through one interval, $C$ through two, and the target
$D$ through three. $D$ is obtained by exact simulation of the measurement
pattern, averaged over eight values of $\theta$, and equals $1$ in the
noiseless limit. Here $c_\sigma\in[0,1]$ is the unknown contribution of the
orthogonal component; since $c_\sigma\ge0$, the product $F\cdot D$
lower-bounds $F_{\rm prep}$ \emph{within this model}, for every $\sigma$. It
is not a bound on the hardware: $D$ omits measurement crosstalk,
readout-induced leakage, and coherent error in the feed-forward-conditioned
operations, and is in that sense itself a ceiling. At a feed-forward latency of $t_{\rm ff}=2\,\mu\mathrm{s}$,
$D=0.938$ --- numerically almost equal to the direct-preparation ceiling
$F_{\rm dir}$ of Table~\ref{tab:hardware}, with which it has no connection,
one being a measured hardware fidelity and the other a computed degradation
of the ideal state --- and the remote state-preparation fidelity obeys
$F_{\rm prep} \ge 0.764$. For a central estimate I take the orthogonal
component to retain half its overlap and to suffer the same
computation-layer degradation as the ideal part, $c_\sigma = D/2$, giving
$F_{\rm prep}\approx0.85$; the alternative convention $c_\sigma=\tfrac12$,
in which $\sigma$ bypasses that degradation, gives $0.86$, so the estimate
is insensitive to the choice at the precision quoted.
Fig.~\ref{fig:usecase_robustness} shows the bound and the
$c_\sigma = D/2$ estimate across latencies. The distributed resource measured in
Sec.~\ref{sec:hardware} therefore already supports remote state
preparation well above the classical fidelity limit, with the dominant
sensitivity being classical feed-forward latency through the target
qubit's idle dephasing.

\begin{figure}[t]
  \centering
  \includegraphics[width=0.8\columnwidth]{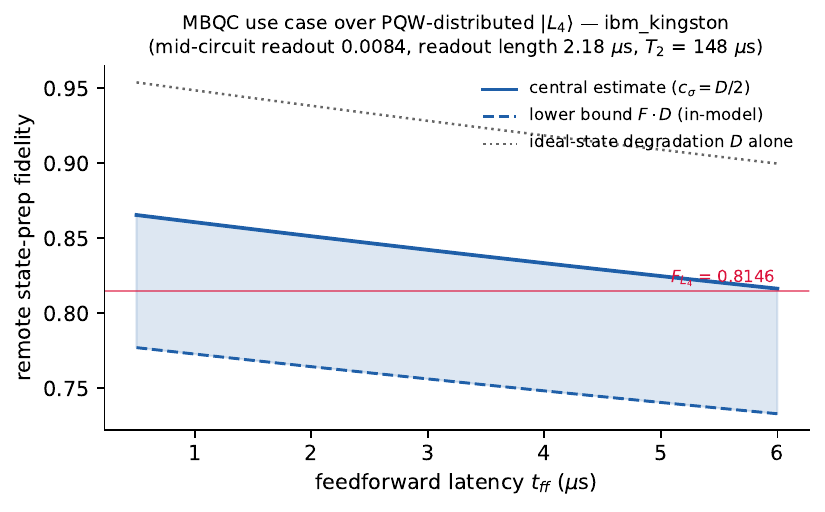}
  \caption{Remote state-preparation fidelity over the
    \texttt{ibm\_kingston}-distributed $\LFour$ state versus classical
    feed-forward latency $t_{\rm ff}$, anchored to the measured
    $F=0.815(3)$ with no noise model imposed. Dashed: the in-model lower
    bound $F\cdot D$. Solid with band: central estimate
    ($c_\sigma=D/2$); calibration values as in
    Table~\ref{tab:calibration}. Dotted: ideal-state degradation $D$ alone,
    from calibrated mid-circuit readout error and idle dephasing.}
  \label{fig:usecase_robustness}
\end{figure}

\paragraph{Significance.}
This use case illustrates the layered architecture enabled by the PQW:
the distribution layer (Stage~1) uses the Universal Z-only Correction
Theorem to establish a shared graph state resource; the computation layer
(Stage~2) consumes that resource via MBQC. The two layers are independent
--- any graph state distributed by Stage~1 can serve as input to a
measurement-based computation in Stage~2 without modification to the
distribution protocol.

\subsection{Open Problems}

(i) \textit{Graph-hardware compatibility function $C(G,H)$.}
Section~\ref{sec:embeddability} gives a \emph{necessary} condition for
native (SWAP-free) execution on a device of coupling girth $\gamma$: the
target graph must have girth at least $\gamma/3$ and maximum degree at most
$3$ (Theorem~\ref{thm:embeddability}), which on heavy-hex ($\gamma=12$)
means triangle-free. $K_4$ violates this and its canonical interaction graph
admits no native embedding on heavy-hex (Corollary~\ref{cor:k4}); $C_4$
saturates it. A
\emph{sufficient} condition, and a quantitative compatibility function
$C(G,H)$ predicting routing overhead when the necessary condition fails,
remain open. Because the constraint is a property of the representative
rather than of the state, a complete account of $C(G,H)$ must also range
over the local-Clifford orbit of the target; the $K_4$/$K_{1,3}$ pair shows
that two representatives of one state can sit on opposite sides of the
embeddability boundary.
Two observations from the hardware campaign suggest that
$C(G,H)$ has an irreducibly \emph{dynamic} component: calibration-based
layout quality metrics only partially predicted measured fidelity, and
qubit coherence drifted by an order of magnitude between calibrations
(Sec.~\ref{sec:hardware}). Selecting embeddings on such hardware is
therefore naturally a sequential decision problem --- candidate native
layouts as arms, measured fidelity as reward, a non-stationary environment
--- for which online (bandit-style) layout selection, informed by learned
predictors of drift propensity from calibration histories, is a promising
route I leave to future work.
The source-target entanglement framework of Prerana and
Wald~\cite{prerana2026entanglement} may inspire one component of such
a function: their matching-based entanglement capacity identifies
structural limits on independent quantum transport channels per graph,
which may constrain the resource efficiency of graph state distribution
on hardware with restricted connectivity.

(ii) \textit{Variational $\mathrm{CP}(\theta)$ walk.}
Replacing $\CZ$ with $\mathrm{CP}(\theta)=\mathrm{diag}(1,1,1,e^{i\theta})$
produces non-stabiliser (magic) states for $\theta\neq 0,\pi$, providing a
framework for distributing both stabiliser and magic state entanglement.

(iii) \textit{Multi-platform hardware validation.}
Extending to IQM Garnet (all-to-all connectivity, CZ-native) will test
the canonical $K_4$ circuit, whose triangles preclude native heavy-hex
embedding (Corollary~\ref{cor:k4}), on hardware without a girth
obstruction, and to
IonQ Forte~1 and Rigetti Ankaa-3
via Amazon Braket~\cite{amazonbraket} will complete the cross-platform
comparison.

\section{Conclusion}
\label{sec:conclusions}

I have presented a protocol that distributes an arbitrary graph state across
a quantum network from $|E|$ elementary two-qubit resources using strictly
local operations and classical communication --- every two-qubit gate acts
within a single party, and no joint operation between separated parties is
required. Each edge realises the entanglement-assisted non-local CZ gate,
and because every gate and correction is diagonal the edges compose in a
single round under the local correction $C_v = Z_v^{g_v}$. A concrete distributed use case
demonstrates that the PQW-distributed $\LFour$ cluster state enables
remote single-qubit state preparation via MBQC, using only single-qubit
measurements and 9 classical bits, with $F=1.0$ verified for all outcomes
(Sec.~\ref{sec:usecase}). The central result is an exact fidelity theory
under independent, qubit-dependent Pauli noise on the resource qubits: noise
acting after the local gate gives a product over vertices weighted by degree,
noise acting on the stored pair gives an edge-cover partition function bounded
above by that product, and dephasing is placement-independent. Exhaustive computational verification across 22 connected graph
topologies confirms $F=1.0$ for all outcomes ($1-F < 10^{-12}$).
On \texttt{ibm\_kingston}, direct fidelity estimation certifies the
distributed states at $F=0.815(3)$, $0.776(3)$, and $0.733(3)$ for
$\LFour$, $\GHZn{4}$, and $\ket{C_4}$ in a single native-layout batch,
with the per-generator spectra exhibiting a degree imprint in which
high-degree nodes carry the dominant error. A girth argument establishes
necessary conditions for native (SWAP-free) execution on a device of
coupling girth $\gamma$ --- the target graph must have girth at least
$\gamma/3$ and maximum degree at most $3$, so on heavy-hex it must be
triangle-free --- placing the canonical interaction graphs $I(P_4)$,
$I(C_4)$ and $I(K_4)$ in a strict embeddability hierarchy, with $I(K_4)$
provably non-embeddable on heavy-hex. That obstruction is a property of the
circuit and not of the state: since $K_4\sim_{\rm LC}K_{1,3}$, the state
$\ket{K_4}$ is reachable from three resource pairs on the natively
embeddable star layout.
The SLOCC-inequivalence of the output states, verified by Schmidt rank,
certifies that the PQW distributes entanglement beyond the GHZ class
reachable by CNOT-shift protocols.

\acknowledgments
I thank Chandan Datta (IISER Kolkata), Tushar (IIT Jodhpur), and
Ambuj (IIT Jodhpur) for valuable discussions and helpful feedback. 

\bibliographystyle{apsrev4-2}
\bibliography{references_paper1}

\end{document}